\documentclass[a4paper, 11pt]{article}

\usepackage{fullpage}
\usepackage{libertine}
\usepackage{amsmath,amssymb,amsthm,mathtools}
\usepackage{enumitem,microtype}
\usepackage[ruled,vlined]{algorithm2e}
\usepackage[numbers,sort&compress]{natbib}

\newtheorem{theorem}{Theorem}[section]

\newtheorem{corollary}[theorem]{Corollary}

\newtheorem{conjecture}[theorem]{Conjecture}

\usepackage{aliascnt}
\newaliascnt{lemma}{theorem}
\newtheorem{lemma}[lemma]{Lemma}
\aliascntresetthe{lemma}

\usepackage[breaklinks=true]{hyperref}
\usepackage[svgnames]{xcolor}
\usepackage[capitalise,nameinlink]{cleveref}
\hypersetup{colorlinks={true},linkcolor={DarkBlue},citecolor=[named]{DarkGreen}}

\crefname{lemma}{Lemma}{Lemmas}
\Crefname{lemma}{Lemma}{Lemmas}

\newcommand{\SC}{\operatorname{SC}}
\newcommand{\MC}{\operatorname{MC}}
\newcommand{\cost}{\operatorname{cost}}
\newcommand{\med}{\operatorname{med}}

\newcommand{\E}{\mathbb{E}}
\newcommand{\w}{\mathbf{w}}
\newcommand{\ov}{\mathbf{o}}
\newcommand{\cv}{\mathbf{c}}
\newcommand{\xv}{\mathbf{x}}
\newcommand{\pv}{\mathbf{p}}
\newcommand{\Mech}[1]{\textnormal{\textsc{#1}}}
\newcommand{\opt}{\textnormal{OPT}}

\SetAlgorithmName{Mechanism}{Mechanism}{List of Mechanisms}
\crefname{algocf}{Mechanism}{Mechanisms}
\Crefname{algocf}{Mechanism}{Mechanisms}
\crefname{algorithm}{Mechanism}{Mechanisms}
\Crefname{algorithm}{Mechanism}{Mechanisms}

\setlist[itemize]{itemsep=2pt,topsep=4pt}
\setlist[enumerate]{itemsep=2pt,topsep=4pt}

\usepackage{authblk}
\usepackage{booktabs}

\title{\bf Truthful Two-Facility Location with Candidate Locations in Metric Spaces}

\author[1]{Panagiotis Kanellopoulos}
\affil[1]{King's College London, UK}
\author[2]{Alexandros A. Voudouris}
\affil[2]{University of Southern Denmark, Denmark}

\date{}

\begin{document}
\maketitle

\begin{abstract}
We study a truthful facility location problem in which agents have private positions in a metric space and known approval preferences over two distinct facilities. Each agent incurs a cost equal to her total distance from the facilities she approves. The goal is to locate the facilities at a given set of candidate locations so as to approximately minimize either the social cost or the max cost, while incentivizing the agents to report their positions truthfully. This problem was previously studied for line metrics and deterministic mechanisms. We first settle the approximation ratio of deterministic mechanisms for the social cost on the line by establishing a lower bound of $7$, matching the upper bound from previous work, and show that randomization leads to a strict improvement. For the max cost, we show that the previously known optimal deterministic guarantee of $3$ remains tight even for universally strategyproof mechanisms. We then turn to general metric spaces. For deterministic strategyproof mechanisms, we obtain tight approximation ratios of $2n-1$ for the social cost and $5$ for the max cost. We also show that randomization yields substantial improvements for both objectives: for the social cost, we obtain a tight characterization, while for the max cost we establish nearly matching upper and lower bounds.
\end{abstract}

\section{Introduction}\label{sec:intro}
In the truthful facility location problem, there is a set of agents with private positions, and one or more facilities that are to be located based on the positions reported by the agents. The goal is to determine the facility locations so that the agents are incentivized to \emph{truthfully} report their positions, while, at the same time, a social function of their individual costs is (approximately) optimized. Since the work of \citet{ProcacciaTennenholtz}, who were the first to consider facility location problems through the prism of \emph{approximate mechanism design without money}, research on this topic has flourished~\citep{ChanSurvey}.

In this paper, we consider the case of two different facilities that can be placed only at different locations chosen from a given finite set of candidates. Each agent has a private position and publicly known approval preferences over the two facilities, and incurs a cost equal to her total distance from the facilities she approves. The requirement that the two facilities be placed at different candidate locations is important even for singleton instances, in which each agent approves exactly one facility: the two location decisions become coupled whenever agents approving different facilities prefer the same candidate location.

This model was previously studied by \citet{KVZ} for line metrics and deterministic strategyproof mechanisms. They showed a tight bound of $3$ for the max cost and an upper bound of $7$ for the social cost, leaving the best possible approximation ratio for the social cost open. Beyond resolving this gap, a natural question is whether randomization can lead to improved approximation ratios. The restriction to the line is a common assumption in the truthful facility location literature, as its linear order provides useful structure for the design and analysis of strategyproof mechanisms. This leads to another natural question: what can be achieved on general metric spaces?

\subsection{Our Contribution}\label{sec:contribution}
We answer the questions above almost completely; an overview of our results is given in \cref{tab:results}. We first complete the picture on the line, for both deterministic and universally strategyproof mechanisms, and then extend the study to arbitrary metric spaces, where we establish nearly tight bounds for both classes of mechanisms.

\begin{itemize}
\item \textbf{The line.}
We settle the remaining open case for deterministic mechanisms on the line by showing that every strategyproof mechanism has approximation ratio at least $7$ in terms of the social cost, even when each agent approves exactly one facility, with different agents possibly approving different facilities. Together with the upper bound of \citet{KVZ}, this shows that the optimal approximation ratio is exactly $7$. We then show that randomization improves the bound to between $5/2$ and $3$ for universally strategyproof mechanisms.\footnote{In \cref{app:line52}, we present a mechanism that we conjecture to achieve an upper bound of $5/2$ in terms of the social cost on the line.} For the max cost, \citet{KVZ} showed that the optimal deterministic approximation ratio is exactly $3$. We show that randomization does not improve this bound: the optimal approximation ratio remains exactly $3$ even for universally strategyproof mechanisms.

\item \textbf{General metric spaces.}
For the social cost, we show that the optimal approximation ratio of deterministic strategyproof mechanisms is exactly $2n-1$. Thus, in contrast to the line, no deterministic strategyproof mechanism can achieve a constant approximation ratio. We then show that randomization improves the optimal approximation ratio to  exactly $4$. For the max cost, the optimal approximation ratio is exactly $5$ for deterministic strategyproof mechanisms, and lies between $15/4$ and $19/5$ for universally strategyproof mechanisms.
\end{itemize}

\begin{table}[t]
\centering
\small
\begin{tabular}{@{}llcc@{}}
\toprule
Metric space & Objective & Deterministic & Universally strategyproof \\
\midrule
Line & Social cost & $7$ (upper bound by \citep{KVZ}) & $[2.5,3]$ \\
Line & Max cost & $3$ (by \citep{KVZ}) & $3$ \\
General & Social cost & $2n-1$ & $4$ (as $n\to\infty$) \\
General & Max cost & $5$ & $[3.75,3.8]$ \\
\bottomrule
\end{tabular}
\caption{An overview of the bounds on the optimal approximation ratio of deterministic and universally strategyproof mechanisms for general instances. Bounds without a reference are shown in this paper.}
\label{tab:results}
\end{table}

The mechanisms and techniques used on the line are quite different from those used on general metric spaces. On the line, our randomized mechanisms build on weighted-percentile rules generalizing the \Mech{Spike} mechanism of \citet{FFG}, and combine these rules to deal with collisions between the two facilities. On general metric spaces, our upper bounds are obtained by randomizing over simple fixed-agent mechanisms. For the lower bounds, we exploit the ability of arbitrary metric spaces to realize unrestricted rankings over candidate locations, allowing us to apply the Gibbard--Satterthwaite theorem.

\subsection{Related Work}\label{sec:related}
Beyond the work of \citet{KVZ} discussed above, several papers study truthful facility location under constraints on the feasible facility locations. \citet{SuiBoutilier} considered constrained facility location more broadly, including settings with multiple homogeneous facilities, where the facilities are identical and each agent cares only about the closest one. \citet{TangEtAl} considered a homogeneous setting with candidate locations on the line; for two facilities, they established approximation bounds for group-strategyproof mechanisms in terms of both the social and the max cost.

For heterogeneous facilities, where the facilities are different and agents may have preferences over which of them affect their cost, \citet{ZhaoMax} considered a constrained $k$-facility model with max-variant cost, where an agent's cost is her distance from the farthest facility she approves. \citet{LotfiVoudouris} studied this two-facility max-variant setting with private positions and publicly known approval preferences. \citet{KanellopoulosVoudourisObnoxious} studied a constrained model with two obnoxious facilities, private positions, and known optional preferences. Another constrained model is the agent-constrained setting of \citet{DeligkasLotfiVoudouris}, where facilities must be placed at distinct locations reported by the agents rather than at a fixed set of candidate locations. \citet{DeligkasKanellopoulosVoudouris} subsequently studied a distributed version of this model, while \citet{XuEtAlDistributed} considered a distributed constrained two-facility model with max-variant cost.

There is also a substantial literature on heterogeneous facility location under different informational and cost models. \citet{SerafinoVentre} considered a discrete setting in which agents occupy nodes of a line graph, their positions are known, their approval preferences are private, and their individual cost is the total distance from all the facilities they approve. \citet{DiscreteHetero} subsequently improved the bounds for both the social and the max cost, while \citet{DiscreteLineBeyond} closed the remaining gap for the social cost on the line and extended the study to general graph metrics. Other works study heterogeneous facilities with optional preferences under different assumptions on individual costs and private information. These models were considered by \citet{ChenOptional} and \citet{ShaOptional}, while \citet{LiHeterogeneous} obtained constant approximation guarantees on general metric spaces.

Finally, truthful multiple-facility location has been studied extensively without candidate-location constraints. In the classical two-facility model, the facilities are homogeneous and each agent incurs her distance from the closest facility. \citet{ProcacciaTennenholtz} initiated the study of this problem, \citet{LuEtAl} showed an asymptotically linear lower bound for deterministic mechanisms and a randomized mechanism with constant approximation ratio on general metric spaces, and \citet{FotakisTzamos} characterized deterministic strategyproof mechanisms on the line and showed that the optimal approximation ratio for the social cost is $n-2$.

\section{Preliminaries}\label{sec:prelim}
We consider the two-facility location problem with candidate locations. An instance $I$ consists of a set $N$ of $n\geq 2$ agents, two facilities $\{F_1,F_2\}$, and a set $C$ of $m\geq 2$ candidate locations. The agents and candidate locations belong to a metric space $(X,d)$, where $d$ satisfies the triangle inequality; on the line, $X=\mathbb{R}$ and $d(x,y)=|x-y|$. Each agent $i\in N$ has a private position $x_i\in X$ and a known approval preference $p_{ij}\in\{0,1\}$ for each $j\in[2]$, indicating whether she approves facility $F_j$, with $p_{i1}+p_{i2}\geq 1$. The candidate set satisfies $C\subseteq X$. We denote an instance by $I=(\xv,\pv,C)$, where $\xv=(x_i)_{i\in N}$ is the position profile and $\pv=(p_{ij})_{i\in N,j\in[2]}$ is the preference profile. For convenience, we will write $d(i,y)$ instead of $d(x_i,y)$ for the distance between agent $i$ and a point $y\in X$.

A \emph{feasible solution} (or simply \emph{solution}) is a pair $\cv=(c_1,c_2)\in C^2$ of distinct candidate locations, where facility $F_j$ is placed at $c_j$ for each $j\in[2]$. The individual cost of an agent $i\in N$ under a solution $\cv$ is the total distance from the facilities she approves:
\[
\cost_i(\cv\mid I)=\sum_{j\in[2]} p_{ij}\cdot d(i,c_j).
\]
When the instance is clear from context, we simply write $\cost_i(\cv)$.
We consider two social objective functions:
\begin{itemize}
    \item The \emph{social cost} of a solution $\cv$ is the total individual cost,
    \[
    \SC(\cv\mid I)=\sum_{i\in N}\cost_i(\cv\mid I).
    \]
    \item The \emph{max cost} of a solution $\cv$ is the maximum individual cost,
    \[
    \MC(\cv\mid I)=\max_{i\in N}\cost_i(\cv\mid I).
    \]
\end{itemize}
Throughout the paper, we will denote by $\w=(w_1,w_2)$ the solution returned by the mechanism under consideration, and by $\ov=(o_1,o_2)$ an optimal solution for the objective under consideration.

A deterministic mechanism $M$ takes as input an instance $I$ and returns a feasible solution $M(I)$. The \emph{approximation ratio} of $M$ with respect to an objective function $f\in\{\SC,\MC\}$ is
\[
\sup_I \frac{f(M(I)\mid I)}{\min_{\cv} f(\cv\mid I)},
\]
where the minimum is over all feasible solutions $\cv$.
The mechanism $M$ is \emph{strategyproof} if no agent can decrease her individual cost by misreporting her position. Formally, for every instance $I=(\xv,\pv,C)$, every agent $i\in N$, and every alternative report $x_i'\in X$,
\[
\cost_i(M(I)\mid I)
\leq
\cost_i\bigl(M((x_i',\xv_{-i}),\pv,C)\mid I\bigr),
\]
where $(x_i',\xv_{-i})$ is the position profile obtained from $\xv$ by replacing only the position of agent $i$ with $x_i'$.
A \emph{randomized} mechanism outputs a probability distribution over feasible solutions. Its approximation ratio is defined as above, by comparing the expected value of the objective with the optimal deterministic value. A randomized mechanism is \emph{universally strategyproof} if it is a probability distribution over deterministic strategyproof mechanisms, where the distribution does not depend on the reported positions.

We now introduce some further useful notation and terminology. 
For each $j\in[2]$, we denote by $N_j$ the set of agents that approve facility $F_j$. An agent in $N_1\cap N_2$ has a \emph{doubleton} preference, while an agent in $N_1\setminus N_2$ or $N_2\setminus N_1$ has a \emph{singleton} preference. Accordingly, an instance is \emph{doubleton} if $N_1\cap N_2=N$, and \emph{singleton} if $N_1\cap N_2=\varnothing$. We write $n_1=|N_1\setminus N_2|$, $n_2=|N_2\setminus N_1|$, $n_{12}=|N_1\cap N_2|$ for the number of agents of each preference type; hence, $n=n_1+n_2+n_{12}$. For each agent $i$, we also denote by $t(i)$ and $s(i)$ her closest and second-closest candidate locations, respectively. All ties are broken according to a fixed order.

\section{The Line}\label{sec:line}
In this section, we consider line metrics. For the social cost, we establish that the optimal approximation ratio is $7$ for deterministic strategyproof mechanisms and between $5/2$ and $3$ for universally strategyproof mechanisms. We then consider the max cost and show that the tight deterministic bound of $3$ due to \citet{KVZ} remains tight for universally strategyproof mechanisms.

\subsection{Social Cost}\label{subsec:line-sc}

\subsubsection{Deterministic Mechanisms}\label{subsec:line-det}

\citet{KVZ} showed an upper bound of $7$ on the approximation ratio of deterministic strategyproof mechanisms in terms of the social cost on the line. Here, we show a matching lower bound using instances consisting only of agents with singleton preferences.\footnote{This also corrects a previously claimed upper bound of $3$ for such singleton instances, which relied on a manipulable mechanism.}

We say that a mechanism is \emph{unanimous} if, whenever all agents are located at the same candidate location $c\in C$, it places the facility at $c$. The following characterization of deterministic unanimous mechanisms for locating a single facility based on the positions of two agents will also be useful later.

\begin{lemma}\label{lem:fixed-phantom}
Let \(C\) be a finite set of candidate locations on the line. Consider a deterministic unanimous strategyproof mechanism \(f\) that places a facility based on the positions of two agents \(i\) and \(j\). Then, there exists \(\tau\in C\), depending only on \(f\), such that \(f\) places the facility at $\med(t(i),t(j),\tau)$ whenever \(t(i)\) and \(t(j)\) are unique and \(t(i)\leq t(j)\).
\end{lemma}

\begin{proof}
Throughout the proof, we consider only positions with a unique closest candidate location.
We establish a few simple properties. 
Suppose first that $t(i)=t(j)=c \in C$. By unanimity, the facility is placed at \(c\) when both agents are located at \(c\). We move the agents, one at a time, to arbitrary positions whose unique closest candidate location remains \(c\). If the outcome changed when an agent moved, that agent could instead report \(c\) and lead to the facility being placed at \(c\), her unique closest candidate location, contradicting strategyproofness. Hence, whenever \(t(i)=t(j)=c\), the facility must be placed at \(c\).

Next, suppose that \(t(i)<t(j)\). We claim that the facility must be placed at a candidate location in $[t(i),t(j)]$. If it were placed strictly to the left of \(t(i)\), agent \(j\) could report a position whose unique closest candidate location is \(t(i)\). By the previous observation, this would place the facility at \(t(i)\), which is strictly closer to her true position, contradicting strategyproofness. The case where the facility is placed strictly to the right of \(t(j)\) is symmetric.

Moreover, once \(t(i)\) and \(t(j)\) are fixed, the outcome does not depend on the precise positions of the agents. Indeed, among the candidate locations in \([t(i),t(j)]\), every position whose unique closest candidate location is \(t(i)\) ranks the locations from left to right, whereas every position whose unique closest candidate location is \(t(j)\) ranks them from right to left. Thus, if two reports of the same agent with the same unique closest candidate location led to different outcomes, then the agent would strictly prefer one of these outcomes and force it by submitting the corresponding report. This would violate strategyproofness. We can therefore write \(f(t(i),t(j))\) for the outcome.

Let $L$ and $R$ denote the leftmost and rightmost candidate locations, respectively. We begin by characterizing the outcome for a fixed $t(j)$. We claim that
\[
f(t(i),t(j))=\max\{t(i),f(L,t(j))\}.
\]
By the interval property, $f(t(i),t(j))\geq t(i)$. Strategyproofness also implies $f(t(i),t(j))\geq f(L,t(j))$; otherwise, when $L$ is $i$'s unique closest location, she could report a position whose unique closest candidate location is $t(i)$ and force an outcome that is farther to the left and that she strictly prefers. We now consider two cases:
\begin{itemize}
\item 
If $t(i)\leq f(L,t(j))$, then $f(t(i),t(j))>f(L,t(j))$ would allow agent $i$ to report a position whose unique closest candidate location is $L$ and obtain the strictly preferred outcome $f(L,t(j))$. Hence,
$f(t(i),t(j))=f(L,t(j))$.

\item If $t(i)>f(L,t(j))$, we proceed from $f(L,t(j))$ to the right over the candidate locations. The claim holds at $f(L,t(j))$ by the previous case. Since $f(L,t(j))\in C$, for every candidate location $t(i)>f(L,t(j))$, the candidate location immediately to its left is at least $f(L,t(j))$. Suppose the claim has already been established for this candidate location. The induction hypothesis implies that the mechanism places the facility there whenever it is the unique closest candidate location of agent $i$. Place agent $i$ sufficiently close to, and to the right of, the midpoint between this candidate location and $t(i)$, so that $t(i)$ is her unique closest candidate location. If $f(t(i),t(j))>t(i)$, then the candidate location immediately to the left of $t(i)$ is strictly closer to her than the actual outcome. By reporting a position whose unique closest candidate location is that location, she could strictly decrease her cost, contradicting strategyproofness. Therefore, $f(t(i),t(j))=t(i)$.
\end{itemize}
Overall, we have established that $f(t(i),t(j))=\max\{t(i),f(L,t(j))\}$.

It remains to characterize $f(L,t(j))$. Let
$\tau=f(L,R)$.
We claim that
\[
f(L,t(j))=\min\{t(j),\tau\}.
\]
The interval property gives $f(L,t(j))\leq t(j)$, while strategyproofness gives $f(L,t(j))\leq\tau$. Indeed, if $f(L,t(j))>\tau$, then agent $j$, when her unique closest candidate location is $R$, could report a position whose unique closest candidate location is $t(j)$ and obtain a strictly preferred outcome.
\begin{itemize}
\item If $t(j)\geq\tau$, then $f(L,t(j))<\tau$ would allow agent $j$ to report a position whose unique closest candidate location is $R$ and obtain $\tau$, which is strictly closer to her true position. Hence, $f(L,t(j))=\tau$.

\item If $t(j)<\tau$, we proceed from $\tau$ to the left over the candidate locations. The claim holds at $\tau$ by the previous case. Since $\tau\in C$, for every candidate location $t(j)<\tau$, the location immediately to its right is at most $\tau$. Suppose the claim has already been established for this candidate location. The induction hypothesis implies that the mechanism places the facility there whenever it is the unique closest candidate location of agent $j$. Place agent $j$ sufficiently close to, and to the left of, the midpoint between $t(j)$ and this location, so that $t(j)$ is her unique closest candidate location. If $f(L,t(j))<t(j)$, then the candidate location immediately to the right of $t(j)$ is strictly closer to her than the actual outcome. By reporting a position whose unique closest candidate location is that location, she could strictly decrease her cost, contradicting strategyproofness. Therefore, $f(L,t(j))=t(j)$.
\end{itemize}
Overall, we have established that $f(L,t(j))=\min\{t(j),\tau\}$.

Combining the two characterizations,
\[
f(t(i),t(j))
=\max\{t(i),\min\{t(j),\tau\}\}
=\med(t(i),t(j),\tau).
\]
All midpoint positions can be chosen strictly on the appropriate side of the midpoint, so that the relevant closest candidate location is unique and the deviation is strictly beneficial.
\end{proof}

We are now ready to prove the lower bound. 

\begin{theorem}\label{thm:lineSClower}
On the line, the approximation ratio of any deterministic strategyproof mechanism in terms of the social cost is at least $7-\delta$, for any $\delta>0$.
\end{theorem}

\begin{proof}
Suppose towards a contradiction that there is a deterministic strategyproof mechanism $M$ with approximation ratio at most $\rho=7-\delta$ for some $\delta>0$. Choose $q>2$ sufficiently close to $2$, and then choose $\alpha$ and $\beta$ such that
\[
\frac12<\alpha<1,\qquad \frac q2<\beta<\frac{q+1}{2},
\]
and
\begin{equation}\label{eq:line-lb-forcing}
\frac{4-\alpha}{\alpha}>\rho,\qquad \frac{\beta+q}{q-\beta}>\rho.
\end{equation}
This is possible since $(4-\alpha)/\alpha$ tends to $7$ as $\alpha$ tends to $1/2$, whereas $(\beta+q)/(q-\beta)$ tends to $(3q+1)/(q-1)$ as $\beta$ tends to $(q+1)/2$, which in turn tends to $7$ as $q$ tends to $2$.

We consider instances with three candidate locations $\{0, 1, q\}$, two agents $i_1$ and $i_2$ that approve only $F_1$, and two agents $j_1$ and $j_2$ that approve only $F_2$. We denote such an instance by the positions $(x_{i_1},x_{i_2} | x_{j_1},x_{j_2})$ of the agents. Consider the following four instances:
\[
J_1=(1,1 | 0,\alpha),\qquad J_1'=(0,\alpha | 1,1),\qquad
J_2=(1,1 | \beta,q),\qquad J_2'=(\beta,q | 1,1).
\]
In $J_1$, the solution $(1,0)$ has social cost $\alpha$ and is optimal, whereas any solution that does not place $F_2$ at $0$ has social cost at least $4-\alpha$. So, by \eqref{eq:line-lb-forcing}, $M$ must place $F_2$ at $0$ when given $J_1$ as input. Symmetrically, $M$ must place $F_1$ at $0$ when given $J_1'$ as input. In $J_2$, the solution $(1,q)$ has social cost $q-\beta$, whereas any solution that does not place $F_2$ at $q$ has social cost at least $\beta+q$. So, by \eqref{eq:line-lb-forcing}, $M$ must place $F_2$ at $q$ when given $J_2$ as input. Symmetrically, $M$ must place $F_1$ at $q$ when given $J_2'$ as input.

Now, consider the instance $I=(1,1|1,1)$. We first show that $M$ does not place $F_1$ at $q$ when given $I$ as input. Suppose otherwise, and consider the instance $I'=(0,1|1,1)$ obtained from $I$ by moving agent $i_1$ from $1$ to $0$. To maintain strategyproofness, $M$ must place $F_1$ at $q$ in $I'$ as well; otherwise, agent $i_1$ would misreport her position in $I$ as $0$ so that $F_1$ moves to $0$ or $1$, decreasing her cost from $q-1$ to at most $1$. Next, consider the instance $I''=(0,0|1,1)$ obtained from $I'$ by moving agent $i_2$ from $1$ to $0$. The unique optimal solution for $I''$ is $(0,1)$ and has social cost $0$, and thus $M$ must place $F_1$ at $0$ to achieve a finite approximation ratio. But then, agent $i_2$ would misreport her position in $I'$ as $0$ to decrease her cost from $q-1$ to $1$, a contradiction. Hence, $M$ does not place $F_1$ at $q$ when given $I$ as input, and, symmetrically, it does not place $F_2$ at $q$ either. Since the facilities must be placed at different locations, we have
\begin{align*}
M(I)\in\{(0,1),(1,0)\}.
\end{align*}
Suppose that $M(I)=(0,1)$. Consider the single-facility mechanism $f$ that, given the positions of agents $j_1$ and $j_2$, places the facility where $M$ places $F_2$ when agents $i_1$ and $i_2$ are both positioned at $1$. Since $M$ is strategyproof and agents $j_1$ and $j_2$ approve only $F_2$, $f$ is a deterministic strategyproof mechanism. The instances $(1,1|0,0)$ and $(1,1|q,q)$ have optimal social cost $0$, and thus, to achieve a finite approximation ratio, $f$ must place the facility at $0$ when both agents are at $0$, and at $q$ when both agents are at $q$. In addition, since $M(I)=(0,1)$, $f$ places the facility at $1$ when both agents are at $1$. So, by \cref{lem:fixed-phantom}, there exists $\tau\in\{0,1,q\}$ such that $f$ places the facility at $\med(t(j_1),t(j_2),\tau)$ whenever $t(j_1)\leq t(j_2)$. In $J_1$, we have $t(j_1)=0$ and $t(j_2)=1$ (since $\alpha>1/2$), and $F_2$ is placed at $0$; hence, $0=\med(0,1,\tau)$, which implies that $\tau=0$. In $J_2$, we have $t(j_1)=1$ (since $1<\beta<(q+1)/2$) and $t(j_2)=q$, and $F_2$ is placed at $q$; hence, $q=\med(1,q,\tau)$, which implies that $\tau=q$, a contradiction. For the remaining case $M(I)=(1,0)$ we can use symmetric arguments for $F_1$ by fixing the positions of agents $j_1$ and $j_2$ at $1$, and using the instances $J_1'$ and $J_2'$. Therefore, the approximation ratio of $M$ cannot be at most $7-\delta$.
\end{proof}

By combining \cref{thm:lineSClower} with the upper bound of \citet{KVZ}, we obtain the following result.

\begin{corollary}
On the line, the best possible approximation ratio of deterministic strategyproof mechanisms in terms of the social cost is exactly $7$.
\end{corollary}

\subsubsection{Universally Strategyproof Mechanisms}\label{subsec:line-percentile}
We next show that the approximation ratio of universally strategyproof mechanisms is between $5/2$ and $3$, thus establishing a clear separation between deterministic and randomized strategyproof mechanisms.

For a single facility with candidate locations, \citet{FFG} introduced \emph{weighted-percentile voting} mechanisms, which select an agent by rank according to a fixed distribution and place the facility at her closest candidate location. Their \Mech{Spike} mechanism achieves the optimal approximation ratio of $2$. Since our problem reduces to this single-facility setting when $N_1=\varnothing$ or $N_2=\varnothing$, we may assume in what follows that $N_1$ and $N_2$ are both non-empty. We use a one-parameter extension of this rule to combine two serial mechanisms with a mechanism that selects both locations based on the doubleton agents. The percentile parameters and mixture probabilities depend only on the publicly known counts $n_1$, $n_2$ and $n_{12}$. Ties between candidate locations are broken to the left, and ties between agents according to a fixed ordering.

We next define the weighted-percentile rule that will be used in our mechanisms. For $\lambda\ge 0$, define the symmetric distribution function
\begin{equation}\label{line3:eq:F}
F_\lambda(q)=
\begin{cases}
\dfrac{(1+2\lambda)q}{2(1+\lambda-q)}, & 0\le q\le 1/2,\\[6pt]
1-F_\lambda(1-q), & 1/2\le q\le 1.
\end{cases}
\end{equation}
For a non-empty group of agents $G\subseteq N$, the weighted-percentile rule selects its $k$-th leftmost agent as representative with probability
\[
F_\lambda(k/|G|)-F_\lambda((k-1)/|G|).
\]
Given a non-empty candidate set $C'\subseteq C$, let $W_{\lambda,G}(C')$ denote the candidate location in $C'$ closest to this representative. For $\lambda=0$, this coincides with the \Mech{Spike} rule.

Our mechanism is a mixture of the following three component mechanisms:
\begin{itemize}
\item $\Mech{Serial-Percentile}_1$. When $n_1>0$, this mechanism places $F_1$ at $w_1=W_{(n_2+n_{12})/n_1,\,N_1\setminus N_2}(C)$ and then places $F_2$ at $W_{0,N_2}(C\setminus\{w_1\})$.
\item $\Mech{Serial-Percentile}_2$. This mechanism is defined symmetrically, with the roles of the two facilities exchanged.
\item \Mech{Pair-Percentile}. When $n_{12}>0$, this mechanism samples a representative from $N_1\cap N_2$ using parameter $\lambda=(n_1+n_2)/(2n_{12})$. It then places the two facilities at her two closest distinct candidate locations and assigns the facility labels uniformly at random.
\end{itemize}
The percentile parameter $\lambda$ used by each serial mechanism is chosen so that the corresponding percentile bound also accounts for the displacement of the first facility, including the additional loss that may arise when it occupies the optimal location of the second facility.

We run these mechanisms with probabilities proportional to weights that depend only on the publicly known counts $n_1,n_2,n_{12}$. The weights are chosen to balance the contributions of the two singleton types, while assigning additional weight to the doubleton-based mechanism when needed to keep the approximation ratio at $3$. In particular, we define
\begin{equation}\label{line3:eq:weights}
\begin{aligned}
\omega_1 &= n_1(n_1+n_{12}),\\
\omega_2 &= n_2(n_2+n_{12}),\\
\omega_{12} &= \frac12 \left[n_{12}(n_1+n_2+2n_{12})-(n_1-n_2)^2\right]_+
\end{aligned}
\end{equation}
where $[x]_+=\max\{0,x\}$. The corresponding probabilities are $\omega_1/\Omega$, $\omega_2/\Omega$ and $\omega_{12}/\Omega$, where $\Omega=\omega_1+\omega_2+\omega_{12}$ is the total weight. See~\cref{mech:line3}.

Note that all weights are non-negative, and $\Omega>0$. Any undefined component has zero weight and is therefore omitted: if $n_j=0$, then $\omega_j=0$, while if $n_{12}=0$, then $\omega_{12}=0$. If both singleton groups are empty, then $\omega_{12}=n_{12}^2>0$, and only \Mech{Pair-Percentile} is used.

\begin{algorithm}[t]
\caption{\Mech{Serial-Pair}}
\label{mech:line3}
\KwIn{Reported positions of agents, with $N_1$ and $N_2$ both non-empty}
\KwOut{Facility locations $\w=(w_1,w_2)$}
Compute the weights in \eqref{line3:eq:weights}\;
Run $\Mech{Serial-Percentile}_1$ with probability $\omega_1/\Omega$, $\Mech{Serial-Percentile}_2$ with probability $\omega_2/\Omega$, and \Mech{Pair-Percentile} with probability $\omega_{12}/\Omega$\;
\end{algorithm}

\begin{theorem}\label{line3:thm:usp}
\Mech{Serial-Pair} is universally strategyproof.
\end{theorem}

\begin{proof}
All random choices are independent of the reported positions. Fix the component mechanism, the sampled ranks, and, for \Mech{Pair-Percentile}, the random assignment of the two selected locations to the two facilities.

We first recall a standard fixed-rank argument. Suppose that each agent reports her position and is associated with her closest candidate location. For any fixed rank $k$, selecting the candidate associated with the $k$-th leftmost agent is strategyproof under single-peaked preferences. Indeed, fix the reports of all other agents. As the reporting agent varies her position, the outcome is her closest candidate location clipped to an interval determined by the other agents. Since her preferences over candidate locations are single-peaked with her true closest candidate location as the peak, reporting truthfully yields her most preferred attainable outcome.

Consider first a serial component. Every agent involved in the first stage approves only the facility placed at that stage. By the fixed-rank argument, she cannot improve its location by misreporting, and she does not care about the second facility. For every agent involved in the second stage, the location of the first facility, and hence the remaining candidate set, is independent of her report. The fixed-rank argument therefore applies to the second stage as well. A singleton agent does not care about the first facility, while for a doubleton agent its contribution to her cost is already fixed.

It remains to consider \Mech{Pair-Percentile}. Write $C=\{c_1<\cdots<c_m\}$. The two closest distinct candidate locations to any position form a consecutive pair. For an agent at position $x$, the difference between her costs for the consecutive pairs $\{c_\ell,c_{\ell+1}\}$ and $\{c_{\ell+1},c_{\ell+2}\}$ is $d(x,c_\ell)-d(x,c_{\ell+2})$, which changes sign at $(c_\ell+c_{\ell+2})/2$. Since these thresholds are increasing in $\ell$, preferences over consecutive pairs are single-peaked, and an agent's most preferred pair is non-decreasing in her position. Thus, selecting the pair associated with an agent of fixed rank is strategyproof by the same fixed-rank argument. Finally, the random assignment of the two selected locations to the facilities does not affect the cost of a doubleton agent, while singleton agents are ignored by this component.
\end{proof}

We next bound the approximation ratio of \Mech{Serial-Pair}. Let $\ov=(o_1,o_2)$ be an optimal solution and, for each $j\in[2]$, let
\[
\opt_j=\sum_{i\in N_j}d(i,o_j).
\]
Thus, $\SC(\ov)=\opt_1+\opt_2$. For the analysis, we further decompose each $\opt_j$ according to the preference types of the agents:
\[
\opt_j^{(j)}=\sum_{i\in N_j\setminus N_{3-j}}d(i,o_j)
\]
and 
\[
\opt_j^{(12)}=\sum_{i\in N_1\cap N_2}d(i,o_j)
\]
so that $\opt_j=\opt_j^{(j)}+\opt_j^{(12)}$. We first establish a weighted-percentile bound and then use it to bound the cost of each component mechanism. For a non-empty group of agents $G$, let $n_G=|G|$ and $\SC_G(t)=\sum_{i\in G}d(i,t)$. For a point $x$ and a non-empty set $S$, let $d(x,S)=\min_{y\in S}d(x,y)$.

\begin{lemma}[Weighted-percentile bounds]\label{line3:lem:percentile}
Let $v$ be the position of a representative sampled from a non-empty group $G$ with parameter $\lambda\ge0$. Then, for every $o\in\mathbb R$,
\begin{equation}\label{line3:eq:representative}
\E\bigl[\SC_G(v)+(1+2\lambda)n_Gd(v,o)\bigr]
\le(2+2\lambda)\SC_G(o).
\end{equation}
Moreover, for every non-empty $C'\subseteq C$ and $W=W_{\lambda,G}(C')$,
\begin{equation}\label{line3:eq:facility}
\E\bigl[\SC_G(W)+\lambda n_Gd(W,o)\bigr]
\le(2+2\lambda)\SC_G(o)+(1+\lambda)n_G d(o,C').
\end{equation}
\end{lemma}

\begin{proof}
We first prove \eqref{line3:eq:representative}. 
Write each distance as an integral over cuts of the line. Fix a cut, and let $q$ be the fraction of agents in $G$ lying on the side of the cut opposite to $o$. By symmetry, the probability that $v$ lies on that side is $F_\lambda(q)$. Dividing the expected contribution of this cut to the left-hand side of \eqref{line3:eq:representative} by $n_G$, we obtain
\[
q(1-F_\lambda(q))+(1-q)F_\lambda(q)+(1+2\lambda)F_\lambda(q)
=q+2(1+\lambda-q)F_\lambda(q).
\]
For $q\le1/2$, this is equal to $(2+2\lambda)q$. For $q\ge1/2$, it is equal to
\[
(2+2\lambda)q-\frac{(1+\lambda)(2q-1)^2}{\lambda+q},
\]
and is therefore at most $(2+2\lambda)q$. Integrating over all cuts proves \eqref{line3:eq:representative}.

We next prove \eqref{line3:eq:facility}. Let $c\in C'$ be a candidate closest to $o$. Since $W$ is the candidate location in $C'$ closest to $v$, $d(v,W)\le d(v,c)\le d(v,o)+d(o,C')$.
Therefore,
\[
\SC_G(W)+\lambda n_Gd(W,o)
\le
\SC_G(v)+(1+2\lambda)n_Gd(v,o)+(1+\lambda)n_Gd(o,C').
\]
Taking expectations and applying \eqref{line3:eq:representative} proves \eqref{line3:eq:facility}.
\end{proof}

\begin{lemma}\label{line3:lem:components}
Assume $n_1\ge n_2$. Each component of \Mech{Serial-Pair} has expected social cost bounded according to the following coefficients:
\[
\renewcommand{\arraystretch}{1.8}
\begin{array}{c|cccc}
&\opt_1^{(1)}&\opt_2^{(2)}&\opt_1^{(12)}&\opt_2^{(12)}\\ \hline
\Mech{Serial-Percentile}_1
&2+\dfrac{2(n_2+n_{12})}{n_1}&2&1&3\\
\Mech{Serial-Percentile}_2
&2&2+\dfrac{2(n_1+n_{12})}{n_2}&3&1\\
\Mech{Pair-Percentile}
&1&1&2+\dfrac{2n_1}{n_{12}}&
2+\dfrac{n_1+3n_2}{2n_{12}}
\end{array}
\]
\end{lemma}

\begin{proof}
We first consider $\Mech{Serial-Percentile}_1$ and condition on the first facility being placed at some candidate location $c\in C$, that is, $w_1=c$. If $c\ne o_2$, then $o_2$ remains available for the second facility, and \eqref{line3:eq:facility} with $\lambda=0$ bounds its expected cost by $2\opt_2$. The $F_1$-singletons contribute $\SC_{N_1\setminus N_2}(c)$. For every doubleton agent $i$, by the triangle inequality,
$d(i,c)\le d(i,o_1)+d(o_1,c)$.
Hence, the cost of the first facility is at most
\[
\SC_{N_1\setminus N_2}(c)+\opt_1^{(12)}+n_{12}d(c,o_1).
\]
Therefore,
\[
\E[\SC(\w)\mid w_1=c]
\le \SC_{N_1\setminus N_2}(c) +n_{12}d(c,o_1) +\opt_1^{(12)} +2\opt_2^{(2)} +2\opt_2^{(12)}.
\]
If $c=o_2$, the $F_1$-singletons again contribute $\SC_{N_1\setminus N_2}(c)$, while the cost of the first facility for the doubleton agents is exactly $\opt_2^{(12)}$. The second facility is selected from $C\setminus\{o_2\}$. Since $o_1$ remains available,
$d(o_2,C\setminus\{o_2\})\le d(o_1,o_2)$.
Therefore, applying \eqref{line3:eq:facility} with $G=N_2$, $C'=C\setminus\{o_2\}$, $o=o_2$, and $\lambda=0$ bounds the expected cost of the second facility by
\[
2\opt_2+(n_2+n_{12})d(o_1,o_2).
\]
Since $c=o_2$, we have $d(o_1,o_2)=d(c,o_1)$, and thus
\[
\E[\SC(\w)\mid w_1=c]
\le \SC_{N_1\setminus N_2}(c) +(n_2+n_{12})d(c,o_1) +2\opt_2^{(2)} +3\opt_2^{(12)}.
\]
In both cases,
\[
\E[\SC(\w)\mid w_1=c]
\le
\SC_{N_1\setminus N_2}(c)
+(n_2+n_{12})d(c,o_1)
+2\opt_2^{(2)}
+\opt_1^{(12)}
+3\opt_2^{(12)}.
\]
We now apply \eqref{line3:eq:facility} with $G=N_1\setminus N_2$, $C'=C$, $o=o_1$, and $\lambda=(n_2+n_{12})/n_1$. Since $n_G=n_1$, we have $\lambda n_G=n_2+n_{12}$, while $\SC_G(o_1)=\opt_1^{(1)}$ and $d(o_1,C)=0$. Taking expectation over the choice of $c=w_1$ therefore gives
\[
\E[\SC(\w)]
\le \left(2+\frac{2(n_2+n_{12})}{n_1}\right)\opt_1^{(1)} +2\opt_2^{(2)} +\opt_1^{(12)} +3\opt_2^{(12)}.
\]
The second row for $\Mech{Serial-Percentile}_2$ follows symmetrically.

We next consider \Mech{Pair-Percentile}. Let $v$ be the selected representative, and let $c$ and $c'$ be her two closest distinct candidate locations. Then, $d(v,c)+d(v,c')\le d(v,o_1)+d(v,o_2)$.
\begin{itemize}
\item {\bf $F_1$-singleton agents.} 
Since the two selected candidate locations are assigned uniformly at random to the two facilities, the expected cost of any such agent $i$ is 
\[
\frac{d(i,c)+d(i,c')}{2}.
\]
By the triangle inequality, $d(i,c)\le d(i,o_1)+d(o_1,v)+d(v,c)$ and $d(i,c')\le d(i,o_1)+d(o_1,v)+d(v,c')$.
Therefore,
\begin{align*}
\frac{d(i,c)+d(i,c')}{2}
&\le d(i,o_1)+d(v,o_1) +\frac{d(v,c)+d(v,c')}{2}\\
&\le d(i,o_1) +\frac{3d(v,o_1)+d(v,o_2)}{2}.
\end{align*}
Summing over the $n_1$ $F_1$-singleton agents, we have that their expected total cost is at most
\[
\opt_1^{(1)} +\frac{n_1}{2}\bigl(3d(v,o_1)+d(v,o_2)\bigr).
\]

\item {\bf $F_2$-singleton agents.} 
Symmetrically to above, the expected total cost of the $F_2$-singletons is at most
\[
\opt_2^{(2)} +\frac{n_2}{2}\bigl(d(v,o_1)+3d(v,o_2)\bigr).
\]

\item {\bf Doubleton agents.} 
For every doubleton agent $i$, the assignment of the two selected candidate locations to the facilities does not affect her cost. By the triangle inequality,
\[
d(i,c)+d(i,c')
\le
2d(i,v)+d(v,c)+d(v,c')
\le
2d(i,v)+d(v,o_1)+d(v,o_2).
\]
Summing over the doubleton agents, we obtain an upper bound of
\[
2\SC_{N_1\cap N_2}(v)
+n_{12}d(v,o_1)
+n_{12}d(v,o_2).
\]
\end{itemize}
Summing these bounds over all agents, conditional on the selected representative $v$, gives us that
\begin{align*}
\E[\SC(\w)\mid v]
&\le \opt_1^{(1)}+\opt_2^{(2)} \\ 
&\quad+2\SC_{N_1\cap N_2}(v)
+\left(n_{12}+\frac{3n_1+n_2}{2}\right)d(v,o_1) +\left(n_{12}+\frac{n_1+3n_2}{2}\right)d(v,o_2).
\end{align*}
Taking expectation over the random choice of $v$, we obtain
\begin{align*}
\E[\SC(\w)]
&\le \opt_1^{(1)}+\opt_2^{(2)} \\
&\quad +\E\left[ 2\SC_{N_1\cap N_2}(v)
+\left(n_{12}+\frac{3n_1+n_2}{2}\right)d(v,o_1)
+\left(n_{12}+\frac{n_1+3n_2}{2}\right)d(v,o_2)
\right].
\end{align*}
Recall that \Mech{Pair-Percentile} uses $\lambda=(n_1+n_2)/(2n_{12})$. Applying \eqref{line3:eq:representative} with $G=N_1\cap N_2$, $o=o_1$ and $n_G=n_{12}$ yields
\[
\E\left[\SC_{N_1\cap N_2}(v) +(n_{12}+n_1+n_2)d(v,o_1)\right]
\le \left(2+\frac{n_1+n_2}{n_{12}}\right)\opt_1^{(12)}.
\]
Similarly, applying \eqref{line3:eq:representative} with $o=o_2$ yields
\[
\E\left[\SC_{N_1\cap N_2}(v)+(n_{12}+n_1+n_2)d(v,o_2) \right]
\le \left(2+\frac{n_1+n_2}{n_{12}}\right)\opt_2^{(12)}.
\]
To obtain the coefficient $n_{12}+\frac{3n_1+n_2}{2}$ for $d(v,o_1)$ and the coefficient $n_{12}+\frac{n_1+3n_2}{2}$ for $d(v,o_2)$, we multiply the first inequality by
\[
1+\frac{n_1-n_2}{2(n_{12}+n_1+n_2)}
\]
and the second by
\[
1-\frac{n_1-n_2}{2(n_{12}+n_1+n_2)}
\]
and then add them. The resulting coefficient of $\opt_1^{(12)}$ is
\[
\left(2+\frac{n_1+n_2}{n_{12}}\right)
\left(1+\frac{n_1-n_2}{2(n_{12}+n_1+n_2)}\right)
= 2+\frac{2n_1}{n_{12}}-\frac{(n_1-n_2)(n_1+n_2)}{2n_{12}(n_{12}+n_1+n_2)}
\le 2+\frac{2n_1}{n_{12}},
\]
where the inequality follows from $n_1\ge n_2$. Similarly, the coefficient of $\opt_2^{(12)}$ is
\[
\left(2+\frac{n_1+n_2}{n_{12}}\right)
\left(1-\frac{n_1-n_2}{2(n_{12}+n_1+n_2)}\right)
= 2+\frac{n_1+3n_2}{2n_{12}} -\frac{n_1-n_2}{2(n_{12}+n_1+n_2)}
\le 2+\frac{n_1+3n_2}{2n_{12}}.
\]
Hence,
\[
\E[\SC(\w)]
\le
\opt_1^{(1)}+\opt_2^{(2)} 
+\left(2+\frac{2n_1}{n_{12}}\right)\opt_1^{(12)} 
+\left(2+\frac{n_1+3n_2}{2n_{12}}\right)\opt_2^{(12)},
\]
and the proof is complete. 
\end{proof}

We are now finally ready to prove the upper bound. 

\begin{theorem}\label{thm:approx-line-3}
On the line, the approximation ratio of \Mech{Serial-Pair} is at most $3$ in terms of the social cost.
\end{theorem}

\begin{proof}
Without loss of generality, assume that $n_1\ge n_2$. We combine the bounds in \cref{line3:lem:components} using the probabilities $\omega_1/\Omega$, $\omega_2/\Omega$, and $\omega_{12}/\Omega$. We bound the resulting coefficients separately.

\begin{itemize}
\item {\bf Coefficient of $\opt_1^{(1)}$.}
By \cref{line3:lem:components}, and using $\omega_1=n_1(n_1+n_{12})$, the coefficient of $\opt_1^{(1)}$ is
\[
\frac{1}{\Omega}\left[ \omega_1\left(2+\frac{2(n_2+n_{12})}{n_1}\right) +2\omega_2+\omega_{12} \right]
= 2-\frac{\omega_{12}}{\Omega} +\frac{2(n_1+n_{12})(n_2+n_{12})}{\Omega}.
\]
By the definition of $\omega_{12}$,
\[
2\omega_{12}
= \left[n_{12}(n_1+n_2+2n_{12})-(n_1-n_2)^2\right]_+
\ge n_{12}(n_1+n_2+2n_{12})-(n_1-n_2)^2.
\]
Equivalently,
\[
2(n_1+n_{12})(n_2+n_{12})
\le \omega_1+\omega_2+2\omega_{12}
= \Omega+\omega_{12}.
\]
Hence,
\[
2-\frac{\omega_{12}}{\Omega} +\frac{2(n_1+n_{12})(n_2+n_{12})}{\Omega} \le3.
\]

\item {\bf Coefficient of $\opt_2^{(2)}$.}
Similarly to above, this coefficient is also at most $3$.

\item {\bf Coefficient of $\opt_1^{(12)}$.}
By \cref{line3:lem:components}, the coefficient of $\opt_1^{(12)}$ is
\[
\frac{1}{\Omega}\left[\omega_1+3\omega_2 +\omega_{12}\left(2+\frac{2n_1}{n_{12}}\right)
\right].
\]
If $\omega_{12}=0$, then $\Omega = \omega_1 + \omega_2$, and the coefficient of $\opt_1^{(12)}$ becomes
\[
\frac{\omega_1+3\omega_2}{\omega_1+\omega_2}\le3.
\]
Suppose now that $\omega_{12}>0$. Then, $n_{12}>0$ and, by definition,
\[
2\omega_{12}
= n_{12}(n_1+n_2+2n_{12})-(n_1-n_2)^2.
\]
Since $(n_1-n_2)^2\ge0$ and $n_1\ge n_2$, we have
\[
\frac{2n_1\omega_{12}}{n_{12}}
\le n_1(n_1+n_2+2n_{12})
\le 2n_1(n_1+n_{12})
= 2\omega_1.
\]
Hence,
\[
\omega_1+3\omega_2 +\omega_{12}\left(2+\frac{2n_1}{n_{12}}\right)
\le 3\omega_1+3\omega_2+2\omega_{12}
\le 3\Omega.
\]
Therefore, the coefficient of $\opt_1^{(12)}$ is at most $3$.

\item {\bf Coefficient of $\opt_2^{(12)}$.}
By \cref{line3:lem:components}, the coefficient of $\opt_2^{(12)}$ is
\[
\frac{1}{\Omega}\left[
3\omega_1+\omega_2
+\omega_{12}\left(2+\frac{n_1+3n_2}{2n_{12}}\right)
\right].
\]
If $\omega_{12}=0$, similarly to the case of $\opt_1^{(12)}$, the coefficient of $\opt_2^{(12)}$ is at most $3$. 
Suppose now that $\omega_{12}>0$. It suffices to show that
\begin{align*}
3\omega_1+\omega_2+\omega_{12}\left(2+\frac{n_1+3n_2}{2n_{12}}\right) \le3\Omega
&\Leftrightarrow
2\omega_2+\omega_{12}\left(1-\frac{n_1+3n_2}{2n_{12}}\right) \ge0 \\
&\Leftrightarrow
8n_{12}\omega_2 +2\omega_{12}(2n_{12}-n_1-3n_2) \ge0.
\end{align*}
Substituting $2\omega_{12} = n_{12}(n_1+n_2+2n_{12})-(n_1-n_2)^2$ and $\omega_2=n_2(n_2+n_{12})$, the left-hand side becomes
\[
2n_2\bigl(2(n_1-n_2-n_{12})^2+(n_1-n_2)n_{12}\bigr) +(n_1-n_2-2n_{12})^2(n_1-n_2+n_{12})
\]
and is non-negative since $n_1 \geq n_2$.
Hence, the coefficient of $\opt_2^{(12)}$ is indeed at most $3$.
\end{itemize}
Since every coefficient is at most $3$, we obtain $\E[\SC(\w)]\le 3\SC(\ov)$, and hence the approximation ratio of \Mech{Serial-Pair} is at most $3$.
\end{proof}

For the lower bound, recall that the proof of \cref{thm:lineSClower} shows that no deterministic strategyproof mechanism can be optimal on four instances with optimal social cost $0$ and, at the same time, have approximation ratio close to $7$ on four further instances. Averaging over the deterministic mechanisms in the support of a universally strategyproof mechanism gives the following bound.

\begin{theorem}\label{thm:line-usp-lb}
On the line, the approximation ratio of any universally strategyproof mechanism in terms of the social cost is at least $5/2-\delta$, for any $\delta>0$.
\end{theorem}
\begin{proof}
Fix $\delta>0$, and consider the construction in the proof of \cref{thm:lineSClower} for a target approximation ratio of $\rho = 7-4\delta$. There are two agents approving only $F_1$, two agents approving only $F_2$, the candidate locations are $\{0,1,q\}$, and $J_1=(1,1|0,\alpha)$, $J_2=(1,1|\beta,q)$, $J_1'=(0,\alpha|1,1)$, $J_2'=(\beta,q|1,1)$, where $q>2$, $\alpha\in(1/2,1)$ and $\beta\in(q/2,(q+1)/2)$ are chosen so that \eqref{eq:line-lb-forcing} holds.

The instances $(1,1|0,0),(1,1|q,q),(0,0|1,1),(q,q|1,1)$ have optimal social cost $0$, so a universally strategyproof mechanism with finite approximation ratio must, with probability $1$, choose a deterministic mechanism that returns an optimal solution on them. Besides strategyproofness, the proof of \cref{thm:lineSClower} uses only optimality on these four instances and an approximation ratio of at most $\rho$ on $J_1,J_2,J_1',J_2'$, and shows that they cannot hold together. Hence, such a mechanism has approximation ratio greater than $7-4\delta$ on at least one of $J_1,J_2,J_1',J_2'$, and at least $1$ on each of the others. Averaging over the random choice of the deterministic mechanism, the expected approximation ratios on $J_1,J_2,J_1',J_2'$ sum to more than $10-4\delta$, so at least one of them exceeds $5/2-\delta$.
\end{proof}

Combining \cref{thm:line-usp-lb} with \cref{thm:approx-line-3}, we obtain the following interval.

\begin{corollary}\label{cor:line-usp-tight}
On the line, the best possible approximation ratio of universally strategyproof mechanisms in terms of the social cost is between $5/2$ and $3$.
\end{corollary}

\subsection{Max Cost}\label{subsec:line-mc}

For the max cost on the line, \citet{KVZ} showed that the optimal approximation ratio of deterministic strategyproof mechanisms is exactly $3$. We show that the same bound remains tight for universally strategyproof mechanisms. It is worth noting that \citet{TangEtAl} established a lower bound of $2$ for the single-facility problem with candidate locations and truthful-in-expectation mechanisms. This lower bound also applies here, since the single-facility problem is a special case of our model when only one facility is approved. 

\begin{theorem}\label{thm:line-usp-mc3}
On the line, the approximation ratio of any universally strategyproof mechanism in terms of the max cost is at least $3-\delta$, for any $\delta>0$.
\end{theorem}

\begin{proof}
Let $m\geq 4$ and $0<\varepsilon<1/2$, and let $C=\{0,1,\ldots,m-1\}$ be the set of candidate locations. Let $i$ and $j$ be two agents that approve only $F_1$. Since no agent approves $F_2$, its location affects neither the max cost nor the incentives of the agents, so every deterministic strategyproof mechanism induces a strategyproof single-facility mechanism for the location of $F_1$. For every $c\in C$, the instance in which both agents are positioned at $c$ has optimal max cost $0$. Hence, a universally strategyproof mechanism with finite approximation ratio must, with probability $1$, choose a deterministic mechanism whose induced single-facility mechanism is unanimous.

For each $\ell\in\{1,\ldots,m-2\}$, let $I_\ell$ be the instance in which agents $i$ and $j$ are positioned at $\ell-\frac12-\varepsilon$ and $\ell+\frac12+\varepsilon$, respectively. Then, $t(i)=\ell-1$, $t(j)=\ell+1$, the optimal location of $F_1$ is $\ell$, and the optimal max cost is $1/2+\varepsilon$. If $F_1$ is placed at $\ell-1$ or $\ell+1$ instead, the farther agent is at distance $3/2+\varepsilon$, so the approximation ratio is
\[
\rho_\varepsilon=\frac{3+2\varepsilon}{1+2\varepsilon}.
\]

Now consider a deterministic mechanism whose induced single-facility mechanism is unanimous. By \cref{lem:fixed-phantom}, there is a fixed $\tau\in C$ such that $F_1$ is placed at $\med(\ell-1,\ell+1,\tau)$ in $I_\ell$. Since $\med(\ell-1,\ell+1,\tau)=\ell$ if and only if $\tau=\ell$, the mechanism returns an optimal solution on at most one of the instances $I_\ell$, and places $F_1$ at $\ell-1$ or $\ell+1$ on all the others. Consequently, its average approximation ratio over the $m-2$ instances is at least
\[
\frac{1+(m-3)\rho_\varepsilon}{m-2}.
\]
Since this holds with probability $1$ for the deterministic mechanism chosen by the universally strategyproof mechanism, averaging over this random choice shows that some instance $I_\ell$ has expected approximation ratio at least the same quantity. Choosing $\varepsilon$ sufficiently small and $m$ sufficiently large, this is at least $3-\delta$.
\end{proof}

\section{General Metric Spaces}\label{sec:general}
We now turn to general metric spaces. We first establish a set of structural properties that will be useful for our lower bounds, and introduce a fixed-agent mechanism that yields both deterministic upper bounds. We then consider the two objectives in turn. For the social cost, the optimal deterministic approximation ratio is $2n-1$, while universally strategyproof mechanisms achieve an asymptotically tight ratio of $4$. For the max cost, the optimal deterministic approximation ratio is $5$, while the optimal ratio for universally strategyproof mechanisms lies between $15/4$ and $19/5$.

\subsection{Structural Properties and a Basic Mechanism}\label{subsec:metric-tools}
We establish two structural properties: arbitrary metric spaces can realize every strict ranking of the candidate locations, and, once all rankings are realized, the location of a facility approved only by agents with singleton preferences is determined by a dictator over its range. We then introduce the fixed-agent mechanism used for both deterministic upper bounds.

\begin{lemma}\label{lem:samerank}
Consider a deterministic strategyproof mechanism, a facility $F_j$, and an agent $i\in N_j\setminus N_{3-j}$, while fixing the positions of all other agents. If two positions of $i$ induce the same strict ranking over the candidate locations, then $F_j$ must be placed at the same location under both reports.
\end{lemma}

\begin{proof}
Suppose that two positions of $i$ induce the same strict ranking over the candidate locations, but lead to different locations of $F_j$. Since $i\in N_j\setminus N_{3-j}$, $i$'s cost depends only on the location of $F_j$. Thus, she strictly prefers one of the two outcomes. If her true position leads to the less preferred outcome, $i$ can report the other position and lead to a more preferred outcome, contradicting strategyproofness. Hence, $F_j$ must be placed at the same location under both reports.
\end{proof}

\begin{lemma}\label{lem:rankings}
Let $(X,d)$ be a finite metric space, and $C\subseteq X$ with $|C|\ge2$. Then, $(X,d)$ can be extended to a metric space $(X',d')$ with $X\subseteq X'$, such that $d'$ coincides with $d$ on $X$ and every strict ranking of $C$ is induced by some point of $X'$.
\end{lemma}

\begin{proof}
Let $D$ be larger than the diameter of $(X,d)$, and
\[
0<\eta<\frac1{|C|}\min_{c\neq c'}d(c,c').
\]
For every strict ranking $\pi$ of $C$, we add a point $v_\pi$ and connect it to each $c\in C$ by an edge of length $D+\eta\cdot\operatorname{rank}_\pi(c)$, where $\operatorname{rank}_\pi(c)\in[|C|]$ is the rank of $c$ in $\pi$. Let $d'$ be the shortest-path metric of the resulting graph, in which any two points of $X$ are joined by an edge of length equal to their distance. Any path between two points of $X$ that passes through a new point has length more than $2D$, so $d'$ coincides with $d$ on $X$. Moreover, for any two distinct candidate locations $c$ and $c'$, we have $\eta|\operatorname{rank}_\pi(c)-\operatorname{rank}_\pi(c')|<d(c,c')$, so the edge between $v_\pi$ and $c$ is a shortest path. Hence, $d'(v_\pi,c)=D+\eta\cdot\operatorname{rank}_\pi(c)$ for every $c\in C$, and $v_\pi$ induces exactly the ranking $\pi$.
\end{proof}

We use the range form of Gibbard's theorem~\citep{Gibbard}: a deterministic strategyproof social choice function on the unrestricted domain of strict rankings whose range contains at least three alternatives has a \emph{dictator on its range}, that is, a fixed agent whose most-preferred alternative among those in the range is selected at every profile.

\begin{lemma}\label{lem:dictator}
Consider a metric space in which every strict ranking of $C$ is induced, a deterministic strategyproof mechanism, and $j\in[2]$ such that $G=N_j\setminus N_{3-j}$ is non-empty. 
Fix the positions of all agents outside $G$, and let $\mathcal R$ be the set of candidate locations at which $F_j$ is placed over profiles in which all agents in $G$ are positioned so as to induce strict rankings of $C$.
If $|\mathcal R|\ge3$, then there exists an agent $i^*\in G$ such that, at every such profile, $F_j$ is placed at the location in $\mathcal R$ closest to $i^*$.
\end{lemma}

\begin{proof}
By \cref{lem:samerank}, applied to one agent at a time, the location of $F_j$ depends only on the rankings induced by the positions of the agents of $G$. Since every strict ranking is induced by some point, this defines a social choice function on the unrestricted domain of strict rankings of $C$, whose range is $\mathcal R$. This function is strategyproof: every agent of $G$ approves only $F_j$, so a ranking that yields a more preferred location corresponds to a position misreport that strictly decreases her cost. The range form of Gibbard's theorem completes the proof.
\end{proof}

\noindent 
{\bf A basic deterministic mechanism.}
We now introduce a simple mechanism, which selects fixed agents as dictators according to their preference types.
\begin{itemize}
\item If there is at least one agent with doubleton preference, the facilities are placed at the two closest candidate locations to a fixed agent $i$ with doubleton preference.

\item Otherwise, if both facilities are approved by at least one agent, $F_1$ is placed at the closest candidate location to a fixed agent $i$ that approves $F_1$, and $F_2$ is placed at the closest available candidate location to a fixed agent $j$ that approves $F_2$.

\item Otherwise, only one facility is approved. This facility is placed at the closest candidate location to a fixed agent $i$, while the other facility is placed at $i$'s second-closest candidate location.
\end{itemize}
We refer to this mechanism as \Mech{Fixed-Dictator}; see \cref{mech:fixed-dictator}.

\begin{algorithm}[t]
\caption{\Mech{Fixed-Dictator}}
\label{mech:fixed-dictator}
\KwIn{Reported positions of agents}
\KwOut{Facility locations $\w=(w_1,w_2)$}
\If{$N_1\cap N_2\neq\varnothing$}{
    $i\leftarrow$ fixed agent in $N_1\cap N_2$\;
    $w_1\leftarrow t(i)$; $w_2\leftarrow s(i)$\;
}
\ElseIf{$N_1\neq\varnothing$ and $N_2\neq\varnothing$}{
    $i\leftarrow$ fixed agent in $N_1$; $j\leftarrow$ fixed agent in $N_2$\;
    $w_1\leftarrow t(i)$\;
    \eIf{$t(j)$ is available}{
        $w_2\leftarrow t(j)$\;
    }{
        $w_2\leftarrow s(j)$\;
    }
}
\Else{
    $k\leftarrow$ the index in $[2]$ such that $N_k\neq\varnothing$\;
    $i\leftarrow$ fixed agent in $N_k$\;
    $w_k\leftarrow t(i)$; $w_{3-k}\leftarrow s(i)$\;
}
\end{algorithm}

\begin{theorem}\label{thm:fixed-dictator-sp}
\Mech{Fixed-Dictator} is strategyproof.
\end{theorem}

\begin{proof}
Since the approval preferences of the agents are known, no agent can affect which of the three cases of the mechanism applies. Moreover, only the fixed agents can affect the outcome, so it suffices to consider their incentives. In the first case, the fixed agent $i$ approves both facilities, and the mechanism places them at her two closest candidate locations. Hence, her individual cost is minimized. In the third case, only one facility is approved, and this facility is placed at the closest candidate location to the fixed agent $i$; the location of the other facility does not affect her cost. Thus, agent $i$ cannot benefit from misreporting in either case.

Consider now the second case, where $N_1\cap N_2=\varnothing$ and both $N_1$ and $N_2$ are non-empty. Agent $i$ approves only $F_1$, which is placed at her closest candidate location. Although her reported position may affect the location of $F_2$, this does not affect her cost. Moreover, the location of $F_1$ is independent of the position of agent $j$. Given the location of $F_1$, agent $j$ approves only $F_2$, which is placed at her closest available candidate location. Hence, neither $i$ nor $j$ can benefit from misreporting, and the mechanism is strategyproof.
\end{proof}

\subsection{Social Cost}\label{subsec:metric-sc}

\subsubsection{Deterministic Mechanisms}

\begin{theorem}\label{thm:gscupper}
On general metric spaces, the approximation ratio of \Mech{Fixed-Dictator} in terms of the social cost is at most $2n-1$.
\end{theorem}

\begin{proof}
Let $\ov=(o_1,o_2)$ be an optimal solution. We consider the three cases of the mechanism separately. In all cases, we use the following simple consequence of the triangle inequality: for any agent $a$ and any $t\in[2]$,
\begin{equation}\label{eq:shift}
d(a,w_t)\leq d(a,o_t)+d(w_t,o_t).
\end{equation}

\medskip
\noindent
\emph{Case 1: $N_1\cap N_2\neq\varnothing$.}
Since $\{w_1,w_2\}=\{t(i),s(i)\}$ minimizes the individual cost of $i$ among all pairs of different candidate locations, we have $\cost_i(\w)\leq\cost_i(\ov)$. By the triangle inequality,
\[
d(w_1,o_1)+d(w_2,o_2)
\leq \big(d(i,w_1)+d(i,o_1)\big)+\big(d(i,w_2)+d(i,o_2)\big)
=\cost_i(\w)+\cost_i(\ov)
\leq 2\cdot\cost_i(\ov).
\]
By \eqref{eq:shift}, we have
\begin{itemize}
    \item $\cost_a(\w) \leq \cost_a(\ov) + d(w_1,o_1)$ for every $a \in N_1 \setminus N_2$;
    \item $\cost_a(\w) \leq \cost_a(\ov) + d(w_2,o_2)$ for every $a \in N_2 \setminus N_1$;
    \item $\cost_a(\w) \leq \cost_a(\ov) + d(w_1,o_1) + d(w_2,o_2)$ for every $a \in N_1 \cap N_2$.
\end{itemize}
Combining all the above with the facts that $n_1+n_{12}-1\leq n-1$, $n_2+n_{12}-1\leq n-1$, and $\cost_i(\ov)\leq\SC(\ov)$, we obtain
\begin{align*}
\SC(\w) 
&= \sum_{a \in N_1 \setminus N_2}\cost_a(\w) + \sum_{a \in N_2 \setminus N_1}\cost_a(\w) + \sum_{a \in (N_1 \cap N_2)\setminus\{i\}}\cost_a(\w) + \cost_i(\w) \\
&\leq \SC(\ov)+(n_1+n_{12}-1)\cdot d(w_1,o_1)+(n_2+n_{12}-1)\cdot d(w_2,o_2)\\
&\leq \SC(\ov)+(n-1)\cdot\big(d(w_1,o_1)+d(w_2,o_2)\big)\\
&\leq \SC(\ov)+2(n-1)\cdot\cost_i(\ov)\\
&\leq (2n-1)\cdot\SC(\ov).
\end{align*}

\medskip
\noindent
\emph{Case 2: $N_1\cap N_2=\varnothing$, and $N_1$ and $N_2$ are both non-empty.}
Since $w_1=t(i)$, we have $d(i,w_1)\leq d(i,o_1)$, and thus, by the triangle inequality,
\[
d(w_1,o_1)\leq d(i,w_1)+d(i,o_1)\leq 2\cdot d(i,o_1).
\]
We distinguish between two subcases depending on whether $w_1=o_2$.
\begin{itemize}
\item $w_1\neq o_2$. Then, $o_2$ is available when $F_2$ is placed, implying that $d(j,w_2)\leq d(j,o_2)$ and, similarly to above, $d(w_2,o_2)\leq 2\cdot d(j,o_2)$. We have that $\cost_i(\w) \leq \cost_i(\ov)$ and $\cost_j(\w) \leq \cost_j(\ov)$. By \eqref{eq:shift}, we also have
\begin{itemize}
    \item For every $a \in N_1\setminus\{i\}$: $\cost_a(\w) \leq d(a,o_1) + d(w_1,o_1) \leq \cost_a(\ov) + 2\cdot d(i,o_1)$. 
    \item For every $a \in N_2\setminus\{j\}$: $\cost_a(\w) \leq d(a,o_2) + d(w_2,o_2) \leq \cost_a(\ov) + 2\cdot d(j,o_2)$. 
\end{itemize}
Since $d(i,o_1)+d(j,o_2)\leq\SC(\ov)$, we obtain
\begin{align*}
    \SC(\w) 
    &= \cost_i(\w) + \sum_{a \in N_1 \setminus\{i\}}\cost_a(\w) + \cost_j(\w) + \sum_{a \in N_2 \setminus\{j\}}\cost_a(\w) \\
    &\leq \SC(\ov)+2(n_1-1)\cdot d(i,o_1)+2(n_2-1)\cdot d(j,o_2) \\
    &\leq (2n-1)\cdot\SC(\ov).
\end{align*}

\item 
$w_1=o_2$. Then, since $w_1=t(i)$, agent $i$ is (weakly) closer to $o_2$ than to $o_1$, and thus, by the triangle inequality,
\[
d(o_1,o_2)\leq d(i,o_1)+d(i,o_2)\leq 2\cdot \cost_i(\ov).
\]
Agent $i$ has individual cost $\cost_i(\w) \leq \cost_i(\ov)$, and, by the triangle inequality, every other agent $a\in N_1$ has individual cost $\cost_a(\w) = d(a,o_2)\leq d(a,o_1)+d(o_1,o_2) \leq d(a, o_1)+2\cdot \cost_i(\ov)$. Hence,
\[
\sum_{a\in N_1} \cost_a(\w) \leq \sum_{a\in N_1} d(a,o_1)+2(n_1-1)\cdot \cost_i(\ov).
\]
Since $o_1\neq o_2=w_1$, the candidate location $o_1$ is available when $F_2$ is placed, and thus
\[
\cost_j(\w) = d(j,w_2)\leq d(j,o_1)\leq d(j,o_2)+d(o_1,o_2) \leq \cost_j(\ov) + 2\cdot \cost_i(\ov).
\]
For every other agent $a\in N_2$, by the triangle inequality and the previous inequality,
\[
\cost_a(\w) = d(a,w_2)\leq d(a,o_2)+d(j,o_2)+d(j,w_2)\leq \cost_a(\ov)+2\cdot \cost_j(\ov)+ 2\cdot \cost_i(\ov).
\]
Hence,
\[
\sum_{a\in N_2} \cost_a(\w)\leq \sum_{a\in N_2} \cost_a(\ov)+2n_2\cdot \cost_i(\ov)+2(n_2-1)\cdot \cost_j(\ov).
\]
By combining these bounds with $n_1+n_2=n$ and $\cost_i(\ov) + \cost_j(\ov) \leq\SC(\ov)$, we obtain
\begin{align*}
\SC(\w)
&\leq \SC(\ov)+2(n-1)\cdot \cost_i(\ov) +2(n_2-1)\cdot \cost_j(\ov) \\
&\leq \SC(\ov)+2(n-1)\cdot\big(\cost_i(\ov) + \cost_j(\ov)\big)\\
&\leq (2n-1)\cdot\SC(\ov).
\end{align*}
\end{itemize}

\medskip
\noindent
\emph{Case 3: all agents approve only $F_k$ for some $k\in[2]$.}
Then, only the location of $F_k$ affects the social cost. Since $w_k=t(i)$, we have $d(w_k,o_k)\leq d(i,w_k)+d(i,o_k)\leq 2\cdot d(i,o_k)$, and agent $i$ has individual cost at most $d(i,o_k)$. By \eqref{eq:shift} for every other agent, we obtain
\[
\SC(\w)\leq \SC(\ov)+2(n-1)\cdot d(i,o_k)\leq(2n-1)\cdot\SC(\ov).
\]

Therefore, the approximation ratio is at most $2n-1$ in every case.
\end{proof}

For the lower bound, we use \cref{lem:rankings,lem:dictator} to show that some agent acts as a dictator, and then position her so that her most preferred location is far from all other agents.

\begin{theorem}\label{thm:gsclower}
On general metric spaces, the approximation ratio of any deterministic strategyproof mechanism in terms of the social cost is at least $2n-1-\delta$, for any $\delta>0$.
\end{theorem}

\begin{proof}
We consider instances in which all agents approve only $F_1$. Let $\varepsilon>0$. Consider three candidate locations $a$, $b$, and $c$ such that
\[
d(a,b)=2,\qquad d(a,c)=3,\qquad d(b,c)=4,
\]
and a point $z$ such that
\[
d(z,a)=1-\varepsilon,\qquad d(z,b)=1+\varepsilon,\qquad d(z,c)=3-\varepsilon.
\]
For sufficiently small $\varepsilon$, these distances define a metric space. Each of the points $a$, $b$, $c$, and $z$ induces a strict ranking over $\{a,b,c\}$; in particular, $z$ induces the ranking $a\succ b\succ c$. By \cref{lem:rankings}, we may add points that induce every strict ranking over $\{a,b,c\}$, without changing the distances above.

Consider now an arbitrary deterministic strategyproof mechanism with finite approximation ratio. When all agents are positioned at $a$, the optimal social cost is $0$, and therefore the mechanism must place $F_1$ at $a$. Similarly, it must place $F_1$ at $b$ when all agents are positioned at $b$, and at $c$ when all agents are positioned at $c$. Hence, the range of the rule that determines the location of $F_1$ is $\{a,b,c\}$, and, by \cref{lem:dictator} with $G=N$, there is an agent $i$ such that $F_1$ is placed at her most preferred candidate location whenever all agents are positioned at points that induce strict rankings.

Consider the profile in which agent $i$ is positioned at $z$ and all other $n-1$ agents are positioned at $b$. Since $i$ ranks $a$ highest, the mechanism places $F_1$ at $a$. The resulting social cost is
\[
\SC(a)=d(z,a)+(n-1)d(a,b)
=(1-\varepsilon)+2(n-1)
=2n-1-\varepsilon.
\]
On the other hand, placing $F_1$ at $b$ gives social cost
\[
\SC(b)=d(z,b)=1+\varepsilon.
\]
Therefore, the approximation ratio of the mechanism is at least
\[
\frac{2n-1-\varepsilon}{1+\varepsilon}.
\]
As $\varepsilon$ tends to $0$, this ratio tends to $2n-1$. Hence, for any $\delta>0$, we can choose $\varepsilon>0$ sufficiently small so that the approximation ratio is at least $2n-1-\delta$.
\end{proof}

By \cref{thm:fixed-dictator-sp,thm:gscupper,thm:gsclower}, we obtain the following result.

\begin{corollary}
On general metric spaces, the best possible approximation ratio of deterministic strategyproof mechanisms in terms of the social cost is exactly $2n-1$.
\end{corollary}

\subsubsection{Universally Strategyproof Mechanisms}\label{subsec:metric-usp-sc}
We now turn to universally strategyproof mechanisms for the social cost. We first present a mechanism with approximation ratio $4$, and then show that this guarantee is asymptotically tight.

\bigskip
\noindent 
{\bf Upper bound.}
If \(N_1=\varnothing\) or \(N_2=\varnothing\), the problem reduces to the single-facility setting, for which 
\Mech{Random Dictatorship} achieves an approximation ratio of at most \(3\) \citep{FFG}. Hence, in the remainder, we assume that \(N_1,N_2\neq\varnothing\). We randomize over the following three component mechanisms:
\begin{itemize}
\item $\Mech{Serial}_1$. 
When $n_1>0$, choose an agent $i\in N_1\setminus N_2$ uniformly at random and place $F_1$ at $w_1=t(i)$. Then choose an agent $j\in N_2$ uniformly at random and place $F_2$ at the location in $C\setminus\{w_1\}$ closest to $j$.

\item $\Mech{Serial}_2$. This mechanism is defined symmetrically, with the roles of the two facilities exchanged.

\item \Mech{Random-Pair}.
When \(n_{12}>0\), choose an agent \(i\in N_1\cap N_2\) uniformly at random, and assign \(F_1\) and \(F_2\) to \(t(i)\) and \(s(i)\) uniformly at random.
\end{itemize}
We run these mechanisms with probabilities proportional to weights that depend only on the publicly known counts $n_1,n_2,n_{12}$. The weights are chosen to balance the contributions of the two singleton types, while
giving additional weight to the doubleton-based mechanism when needed to keep the approximation ratio at $4$. In particular, define
\begin{equation}\label{eq:mixed-dictator-weights}
\begin{aligned}
\omega_1&=n_1(n_1+2n_{12}),\\
\omega_2&=n_2(n_2+2n_{12}),\\
\omega_{12}
&=\frac13\left[
2n_{12}(n+3n_{12})-(n_1-n_2)^2
\right]_+,
\end{aligned}
\end{equation}
where, as before, $[x]_+=\max\{0,x\}$. The corresponding probabilities are $\omega_1/\Omega$, $\omega_2/\Omega$, and $\omega_{12}/\Omega$ with $\Omega=\omega_1+\omega_2+\omega_{12}$. See \cref{mech:mixed-dictator}.

All weights are non-negative, and any inapplicable mechanism receives zero weight: if \(n_1=0\), then \(\omega_1=0\); if \(n_2=0\), then \(\omega_2=0\); and if \(n_{12}=0\), then \(\omega_{12}=0\). Moreover, when both facilities are approved, \(\Omega>0\).

\begin{algorithm}[t]
\caption{\Mech{Mixed-Dictator}}
\label{mech:mixed-dictator}
\KwIn{Reported positions of agents with $N_1$ and $N_2$ both non-empty}
\KwOut{Facility locations $\w=(w_1,w_2)$}
Compute the weights in \eqref{eq:mixed-dictator-weights}\;
Run $\Mech{Serial}_1$ with probability $\omega_1/\Omega$, $\Mech{Serial}_2$ with probability $\omega_2/\Omega$, and \Mech{Random-Pair} with probability $\omega_{12}/\Omega$\;
\end{algorithm}

\begin{theorem}\label{thm:mixed-dictator-usp}
\Mech{Mixed-Dictator} is universally strategyproof.
\end{theorem}

\begin{proof}
We will show that every deterministic realization of \Mech{Mixed-Dictator} is strategyproof. To this end, fix a component mechanism and all of its random choices, including the sampled agents and, for \Mech{Random-Pair}, the assignment of the two selected locations to the facilities.

Consider first $\Mech{Serial}_1$. The agent selected in the first stage approves only $F_1$, which is placed at her closest candidate location. Although her report may affect the candidate set available in the second stage, she does not care about the location of $F_2$. For the agent selected in the second stage, the location of $F_1$, and hence the set of available candidate locations, is independent of her report. The mechanism places $F_2$ at her closest available candidate location. If she has doubleton preference, her cost from $F_1$ is already fixed. Hence, neither selected agent can benefit from misreporting. All other agents are ignored and therefore cannot affect the outcome. The same argument applies symmetrically to $\Mech{Serial}_2$.

In \Mech{Random-Pair}, the selected agent has doubleton preference, and the mechanism chooses her two closest distinct candidate locations. These minimize her total cost among all pairs of distinct candidate locations. Since she approves both facilities, the assignment of the two locations to $F_1$ and $F_2$ does not affect her cost. Again, all other agents are ignored.
\end{proof}

We next bound the approximation ratio. As in the analysis for the line, for an optimal solution $\ov=(o_1,o_2)$, we write the optimal social cost as 
\[
\SC(\ov)=\opt_1+\opt_2, 
\]
where 
\[
\opt_j=\sum_{i\in N_j}d(i,o_j)
\]
is the contribution of facility $F_j$, decomposed into the costs of the singleton and doubleton agents. We start by establishing a simple representative-agent bound and then use it to bound the cost of each component mechanism.

\begin{lemma}
\label{lem:random-representative}
Let $A$ be a non-empty set of candidate locations, let $y\in A$, and let $G$ be a non-empty set of agents. Choose an agent $i\in G$ uniformly at random, and let $\hat y$ be her closest location in $A$. If $A$ and $y$ are fixed independently of the choice of $i$, then, for every set $B$ of agents,
\begin{equation}\label{eq:random-representative}
\E_i\left[\sum_{k\in B}d(k,\hat y)\right]
\le
\sum_{k\in B}d(k,y)
+\frac{2|B|}{|G|}
\sum_{g\in G}d(g,y).
\end{equation}
\end{lemma}

\begin{proof}
Since $d(i,\hat y)\le d(i,y)$, by the triangle inequality, we have that
\[
d(k,\hat y)
\le d(k,y)+d(y,i)+d(i,\hat y)
\le d(k,y)+2d(i,y).
\]
Summing over $k\in B$ and averaging over the uniformly selected $i\in G$ gives us \eqref{eq:random-representative}.
\end{proof}

We next bound the coefficients for the three component mechanisms.

\begin{lemma}\label{lem:metric-sc-components}
Each component of \Mech{Mixed-Dictator} has expected social cost bounded
according to the following coefficients:
\[
\renewcommand{\arraystretch}{1.8}
\begin{array}{c|cccc}
&\opt_1^{(1)}&\opt_2^{(2)}
&\opt_1^{(12)}&\opt_2^{(12)}\\ \hline
\Mech{Serial}_1
&3+\dfrac{2n_2+4n_{12}}{n_1}&3&1&3\\
\Mech{Serial}_2
&3&3+\dfrac{2n_1+4n_{12}}{n_2}&3&1\\
\Mech{Random-Pair}
&1&1&
3+\dfrac{3n_1+n_2}{2n_{12}}&
3+\dfrac{n_1+3n_2}{2n_{12}}
\end{array}
\]
\end{lemma}

\begin{proof}
We first consider $\Mech{Serial}_1$. Let $i\in N_1\setminus N_2$ be the first selected agent and $j\in N_2$ the second one. For the contribution of $F_1$, we apply \eqref{eq:random-representative} with $A=C$, $y=o_1$, $G=N_1\setminus N_2$ and $B=N_1$.
Since
\[
\sum_{k\in N_1}d(k,o_1) = \opt_1^{(1)}+\opt_1^{(12)}, \
|N_1|=n_1+n_{12}, \
\text{ and } \
\sum_{g\in N_1\setminus N_2}d(g,o_1) = \opt_1^{(1)},
\]
we obtain
\[
\E\left[\sum_{k\in N_1}d(k,w_1)\right]
\le
\left(3+\frac{2n_{12}}{n_1}\right)\opt_1^{(1)}
+\opt_1^{(12)}.
\]
The only additional issue for $F_2$ is that the first facility may occupy
$o_2$. We claim that, conditional on $i$ and $j$, $d(j,w_2) \le d(j,o_2)+2d(i,o_1)$.
If $w_1\neq o_2$, then $o_2$ remains available, and the stronger
inequality $d(j,w_2)\le d(j,o_2)$ holds. Otherwise,
$w_1=t(i)=o_2$, and thus $o_1$ remains available. By the triangle inequality, we have
\[
d(j,w_2)
\le d(j,o_1)
\le d(j,o_2)+d(o_1,o_2)
\le d(j,o_2)+d(i,o_1)+d(i,o_2) 
\le d(j,o_2)+2d(i,o_1),
\]
where the last inequality follows from $d(i,o_2)\le d(i,o_1)$.
Consequently, for every $k\in N_2$, using the triangle inequality, 
\[
d(k,w_2)
\le d(k,o_2) + d(j,o_2) + d(j,w_2) 
\le d(k,o_2)+2d(j,o_2)+2d(i,o_1).
\]
Summing over $N_2$ and taking expectation, using
\[
\E[d(j,o_2)] = \frac{\opt_2^{(2)}+\opt_2^{(12)}}{n_2+n_{12}} \
\text{ and } \
\E[d(i,o_1)] = \frac{\opt_1^{(1)}}{n_1},
\]
gives us that
\[
\E\left[\sum_{k\in N_2}d(k,w_2)\right]
\le 3\opt_2^{(2)} +3\opt_2^{(12)} +\frac{2(n_2+n_{12})}{n_1}\opt_1^{(1)}.
\]
By adding the two facility contributions, we have
\[
\E[\SC(\Mech{Serial}_1)]
\le \left(3+\frac{2n_2+4n_{12}}{n_1}\right)\opt_1^{(1)} +3\opt_2^{(2)} +\opt_1^{(12)} +3\opt_2^{(12)}.
\]
The second row for $\Mech{Serial}_2$ follows symmetrically.

We next consider \Mech{Random-Pair}. Let $i\in N_1\cap N_2$ be the selected agent.
Since $t(i)$ and $s(i)$ are her two closest distinct candidate locations,
$d(i,t(i))+d(i,s(i)) \le d(i,o_1)+d(i,o_2)$.
By the triangle inequality,
\[
d(o_1,t(i))+d(o_1,s(i)) \le 3d(i,o_1)+d(i,o_2)
\]
and, symmetrically,
\[
d(o_2,t(i))+d(o_2,s(i)) \le d(i,o_1)+3d(i,o_2).
\]
After averaging over the two equally likely assignments of $t(i)$ and $s(i)$ to $F_1$ and $F_2$, the contribution of the $F_1$-singleton agents is at most
\[
\opt_1^{(1)} +\frac{3n_1}{2}d(i,o_1) +\frac{n_1}{2}d(i,o_2),
\]
while the $F_2$-singleton agents contribute at most
\[
\opt_2^{(2)} +\frac{n_2}{2}d(i,o_1) +\frac{3n_2}{2}d(i,o_2).
\]
Hence, the total contribution of the singleton agents is at most
\[
\opt_1^{(1)}+\opt_2^{(2)} +\frac{3n_1+n_2}{2}d(i,o_1) +\frac{n_1+3n_2}{2}d(i,o_2).
\]
For every doubleton agent $k$, by the triangle inequality,
\[
\cost_k(\w) \le \cost_k(\ov) +d(i,o_1)+d(i,w_1) +d(i,o_2)+d(i,w_2).
\]
Since $w_1$ and $w_2$ are the two closest distinct candidate locations to $i$, $d(i,w_1)+d(i,w_2) \le d(i,o_1)+d(i,o_2)$, and thus 
\[
\cost_k(\w) \le \cost_k(\ov) +2d(i,o_1)+2d(i,o_2).
\]
Overall, the doubleton agents contribute at most
\[
\opt_1^{(12)}+\opt_2^{(12)} +2n_{12}\bigl(d(i,o_1)+d(i,o_2)\bigr).
\]
Combining the two bounds, conditional on $i$,
\begin{align*}
\E[\SC(\w)\mid i]
&\le \opt_1^{(1)}+\opt_2^{(2)} +\opt_1^{(12)}+\opt_2^{(12)}\\
&+\left( 2n_{12}+\frac{3n_1+n_2}{2} \right)d(i,o_1)\\
&+\left(2n_{12}+\frac{n_1+3n_2}{2} \right)d(i,o_2).
\end{align*}
Since $i$ is chosen uniformly from $N_1\cap N_2$, we have that 
\[
\E[d(i,o_1)] = \frac{\opt_1^{(12)}}{n_{12}} \
\text{ and }
\E[d(i,o_2)] = \frac{\opt_2^{(12)}}{n_{12}}.
\]
Therefore,
\[
\E[\SC(\Mech{Random-Pair})]
\le \opt_1^{(1)}+\opt_2^{(2)} 
+ \left( 3+\frac{3n_1+n_2}{2n_{12}} \right)\opt_1^{(12)}
+ \left( 3+\frac{n_1+3n_2}{2n_{12}} \right)\opt_2^{(12)}
\]
and the proof is complete.
\end{proof}

We are now ready to prove the upper bound.

\begin{theorem}\label{thm:metric-usp-sc4}
On general metric spaces, the approximation ratio of
\Mech{Mixed-Dictator} in terms of the social cost is at most $4$.
\end{theorem}

\begin{proof}
We will combine the bounds in \cref{lem:metric-sc-components} using the probabilities $\omega_1/\Omega$, $\omega_2/\Omega$, and $\omega_{12}/\Omega$, and bound the resulting coefficients separately. 
Since there is a symmetry between the coefficients of $\opt_1^{(1)}$ and $\opt_2^{(2)}$, and between the
coefficients of $\opt_1^{(12)}$ and $\opt_2^{(12)}$, it suffices to bound the coefficients of $\opt_1^{(1)}$ and $\opt_1^{(12)}$. 
\begin{itemize}
\item {\bf Coefficient of $\opt_1^{(1)}$.}
If $n_1=0$, then $\opt_1^{(1)}=0$. Otherwise, by
\cref{lem:metric-sc-components}, the coefficient is
\[
\frac{1}{\Omega}\left[\omega_1\left(3+\frac{2n_2+4n_{12}}{n_1}\right)+3\omega_2+\omega_{12}\right].
\]
Using $\omega_1=n_1(n_1+2n_{12})$, it becomes
\[
\frac{3(\omega_1+\omega_2)+2(n_1+2n_{12})(n_2+2n_{12})+\omega_{12}}{\Omega}.
\]
To show that this coefficient is at most $4$, using $\Omega = \omega_1 + \omega_2+\omega_{12}$, we need to equivalently show that 
\[
2(n_1+2n_{12})(n_2+2n_{12}) \le \omega_1+\omega_2+3\omega_{12}.
\]
By the definitions of $\omega_1$ and $\omega_2$, the difference between the first term and the first two terms on the right-hand side is
\begin{align*}
&2(n_1+2n_{12})(n_2+2n_{12})-\omega_1-\omega_2\\
&\qquad=
2n_{12}(n_1+n_2+4n_{12})-(n_1-n_2)^2\\
&\qquad=
2n_{12}(n+3n_{12})-(n_1-n_2)^2.
\end{align*}
By the definition of $\omega_{12}$,
\[
3\omega_{12} = \left[2n_{12}(n+3n_{12})-(n_1-n_2)^2\right]_+
\geq  2n_{12}(n+3n_{12})-(n_1-n_2)^2.
\]
Therefore,
\[
2(n_1+2n_{12})(n_2+2n_{12}) \le \omega_1+\omega_2+3\omega_{12},
\]
as required.

\item {\bf Coefficient of $\opt_1^{(12)}$.}
If $n_{12}=0$, then $\opt_1^{(12)}=0$.
In case $n_{12}>0$ and $\omega_{12}=0$, the coefficient is
\[
\frac{\omega_1+3\omega_2}{\omega_1+\omega_2} \le3.
\]

Suppose now that $\omega_{12}>0$. By
\cref{lem:metric-sc-components}, the coefficient is
\[
\frac{1}{\Omega}\bigg[\omega_1+3\omega_2+\left(3+\frac{3n_1+n_2}{2n_{12}}\right)\omega_{12}\bigg]
\]
To show that the coefficient is at most $4$, we equivalently want
\[
3\omega_1+\omega_2 \ge \frac{3n_1+n_2-2n_{12}}{2n_{12}}\omega_{12}.
\]
This is immediately true when $3n_1+n_2\le2n_{12}$, so suppose that $3n_1+n_2>2n_{12}$. 
Since $\omega_{12}>0$, by definition,
\[
\omega_{12} = \frac13 ( 2n_{12}(n_1+n_2+4n_{12})-(n_1-n_2)^2)
\le \frac23 n_{12}(n_1+n_2+4n_{12}).
\]
Therefore, it suffices to show that
\begin{align*}
&3n_1(n_1+2n_{12})+n_2(n_2+2n_{12})
\ge \frac13 (3n_1+n_2-2n_{12})(n_1+n_2+4n_{12}) \\
&\Leftrightarrow
\frac23
\left(2n_1^2+(n_1-n_2)^2+4n_1n_{12}+2n_2n_{12}+4n_{12}^2\right)\ge0,
\end{align*}
which is true. 
\end{itemize}
Since every coefficient is at most $4$, we obtain $\E[\SC(\w)] \le 4\SC(\ov)$, and hence the approximation ratio of \Mech{Mixed-Dictator} is at most $4$.
\end{proof}

\bigskip
\noindent 
{\bf Lower bound.}
We next show that the upper bound of $4$ is asymptotically tight. The proof
uses a star metric together with \cref{lem:rankings,lem:dictator}.

\begin{theorem}\label{thm:usp-sc-lb4}
On general metric spaces with $n=2k$ agents, the approximation ratio of any
universally strategyproof mechanism in terms of the social cost is at least
$4-4/n-\delta$, for any $\delta>0$.
\end{theorem}

\begin{proof}
Let $|N_1|=|N_2|=k$, with all agents having singleton preferences, and fix
$\varepsilon>0$. Let $c_0,c_1,c_2,c_3$ be the candidate locations, placed at
the leaves of a star with center $z^*$, at distances $1$, $1+\varepsilon$,
$1+2\varepsilon$, and $1+3\varepsilon$ from $z^*$, respectively. Thus, the
distance between two distinct candidate locations is the sum of their
distances from $z^*$. In particular, both $z^*$ and $c_0$ induce the strict
ranking
\[
c_0\succ c_1\succ c_2\succ c_3,
\]
and each of $c_1$, $c_2$, and $c_3$ also induces a strict ranking. By
\cref{lem:rankings}, we may add points so that every strict ranking of the four
candidate locations is induced by some point, without changing these distances.

For each $\ell\in\{1,2,3\}$, consider the instance in which all agents of
$N_1$ are positioned at $c_\ell$ and all agents of $N_2$ are positioned at
$c_0$, together with the symmetric instance obtained by exchanging the two
groups. Each of these six instances has a unique solution with social cost
$0$. Hence, every deterministic mechanism in the support of a universally
strategyproof mechanism with finite approximation ratio must return this
solution on all six instances. Fix such a deterministic mechanism $M$.

Fix the positions of all agents of $N_2$ at $c_0$, and consider the rule that
determines the location of $F_1$. Its range $\mathcal R_1$ contains $c_1$,
$c_2$, and $c_3$, since, when all agents of $N_1$ are positioned at $c_\ell$,
the mechanism must place $F_1$ at $c_\ell$. Hence, by \cref{lem:dictator},
there is an agent $d_1\in N_1$ such that, whenever the agents of $N_1$ are
positioned at points that induce strict rankings, $F_1$ is placed at the
location of $\mathcal R_1$ that $d_1$ ranks highest. Symmetrically, fixing all
agents of $N_1$ at $c_0$ yields a range
$\mathcal R_2\supseteq\{c_1,c_2,c_3\}$ and a dictator $d_2\in N_2$ for the
location of $F_2$.

Consider the instance in which all agents are positioned at $c_0$. Since every
agent has the ranking of $c_0$, each rule selects $c_0$ if it belongs to its
range, and $c_1$ otherwise. If both ranges contained $c_0$, both facilities
would be placed at $c_0$; if neither contained $c_0$, both would be placed at
$c_1$. Since the facilities must be placed at distinct locations, exactly one
of $\mathcal R_1$ and $\mathcal R_2$ contains $c_0$. Let $r\in[2]$ be such
that $c_0\in\mathcal R_r$.

For $g\in[2]$ and $x\in N_g$, let $I_x^g$ be the instance in which all agents
of $N_{3-g}$ are positioned at $c_0$, agent $x$ is positioned at $z^*$, and
all other agents of $N_g$ are positioned at $c_2$, and let
\[
\mathcal I=\{I_x^g:g\in[2],\ x\in N_g\}.
\]
This is a fixed family of $2k$ instances, independent of $M$. For every
$I_x^g\in\mathcal I$, placing $F_{3-g}$ at $c_0$ and $F_g$ at $c_2$ leaves
only agent $x$ with positive cost, equal to
$d(z^*,c_2)=1+2\varepsilon$. Hence,
\[
\SC(\ov_{I_x^g}\mid I_x^g)\le1+2\varepsilon.
\]
Moreover, every solution has social cost at least $1$ on $I_x^g$, since agent
$x$ is at distance at least $1$ from every candidate location.

Consider first $I^r_{d_r}$. Since $d_r$ is positioned at $z^*$, she has the
ranking of $c_0$, and $c_0\in\mathcal R_r$ implies that $F_r$ is placed at
$c_0$. The other facility must be placed at a different location, so each of
the $k$ agents of $N_{3-r}$ at $c_0$ incurs cost at least
$d(c_0,c_1)=2+\varepsilon$. Each of the $k-1$ agents of $N_r$ at $c_2$
incurs cost $d(c_2,c_0)=2+2\varepsilon$, while $d_r$ incurs cost
$d(z^*,c_0)=1$. Therefore,
\[
\SC(M(I^r_{d_r})\mid I^r_{d_r})\ge4k-1.
\]

Next, consider $I^{3-r}_{d_{3-r}}$. Again, $d_{3-r}$ is positioned at $z^*$
and has the ranking of $c_0$. Since $c_0\notin\mathcal R_{3-r}$ and
$c_1\in\mathcal R_{3-r}$, the mechanism places $F_{3-r}$ at $c_1$. Each of
the $k-1$ agents of $N_{3-r}$ at $c_2$ incurs cost
$d(c_2,c_1)=2+3\varepsilon$, while $d_{3-r}$ incurs cost
$d(z^*,c_1)=1+\varepsilon$. Therefore,
\[
\SC(M(I^{3-r}_{d_{3-r}})\mid I^{3-r}_{d_{3-r}})\ge2k-1.
\]

Choose an instance uniformly at random from $\mathcal I$. Besides the two
instances above, each of the remaining $2k-2$ instances has social cost at
least $1$. Therefore,
\[
\E_{I\sim\mathcal I}[\SC(M(I)\mid I)]
\ge\frac{(4k-1)+(2k-1)+(2k-2)}{2k}
=4-\frac2k.
\]
This bound holds for every deterministic mechanism in the support, while
$\mathcal I$ is independent of that mechanism. Averaging over the random
choice of the deterministic mechanism, some instance in $\mathcal I$ has
expected social cost at least $4-2/k$. Since its optimal social cost is at
most $1+2\varepsilon$, its approximation ratio is at least
\[
\frac{4-2/k}{1+2\varepsilon}=\frac{4-4/n}{1+2\varepsilon}.
\]
Choosing $\varepsilon$ sufficiently small completes the proof.
\end{proof}

\begin{corollary}\label{cor:metric-usp-sc}
For even $n$, the best possible approximation ratio of universally
strategyproof mechanisms in terms of the social cost on general metric spaces
lies in $[4-4/n,4]$, and tends to $4$ as $n$ grows.
\end{corollary}

\subsection{Max Cost}\label{subsec:metric-mc}

\subsubsection{Deterministic Mechanisms}

We first consider deterministic mechanisms for the max cost and establish a tight bound of $5$. The upper bound is achieved by \Mech{Fixed-Dictator}.

\begin{theorem}\label{thm:gmcupper}
On general metric spaces, the approximation ratio of \Mech{Fixed-Dictator} in terms of the max cost is at most $5$.
\end{theorem}

\begin{proof}
Let $\ov=(o_1,o_2)$ be an optimal solution. We consider the three cases of the mechanism separately. 

\medskip
\noindent
\emph{Case 1: $N_1\cap N_2\neq\varnothing$.}
Let $i$ be the fixed agent that determines the outcome. Since $\{w_1,w_2\}=\{t(i),s(i)\}$ minimizes the individual cost of $i$ among all pairs of distinct candidate locations,
\[
\cost_i(\w)\leq \cost_i(\ov)\leq \MC(\ov).
\]
By the triangle inequality,
\[
\begin{aligned}
d(w_1,o_1)+d(w_2,o_2)
&\leq d(i,w_1)+d(i,o_1)+d(i,w_2)+d(i,o_2)\\
&=\cost_i(\w)+\cost_i(\ov)\\
&\leq 2\cdot\MC(\ov).
\end{aligned}
\]
We now consider the different types of agents:
\begin{itemize}
    \item For any agent $a\in N_1\setminus N_2$,
\[
\cost_a(\w)
=d(a,w_1)
\leq d(a,o_1)+d(w_1,o_1)
\leq 3\cdot\MC(\ov).
\]

\item For any agent $a\in N_2\setminus N_1$, the same argument as above implies that $\cost_a(\w)\leq 3\cdot\MC(\ov)$.

\item For any agent $a\in N_1\cap N_2$,
\[
\begin{aligned}
\cost_a(\w)
&=d(a,w_1)+d(a,w_2)\\
&\leq d(a,o_1)+d(a,o_2)+d(w_1,o_1)+d(w_2,o_2)\\
&\leq \cost_a(\ov)+2\cdot\MC(\ov)\\
&\leq 3\cdot\MC(\ov).
\end{aligned}
\]
\end{itemize}
Hence, $\MC(\w)\leq 3\cdot\MC(\ov)$ in this case.

\medskip
\noindent
\emph{Case 2: $N_1\cap N_2=\varnothing$, and $N_1$ and $N_2$ are both non-empty.}
Let $i\in N_1$ and $j\in N_2$ be the fixed agents that determine the outcome. Since $w_1=t(i)$,
$d(i,w_1)\leq d(i,o_1)$,
and therefore
\[
d(w_1,o_1)
\leq d(i,w_1)+d(i,o_1)
\leq 2\cdot d(i,o_1)
\leq 2\cdot\MC(\ov).
\]
We distinguish between two subcases depending on whether $w_1=o_2$.
\begin{itemize}
\item $w_1\neq o_2$. Then $o_2$ is available when $F_2$ is placed. Since $w_2$ is the closest available candidate location to $j$,
$d(j,w_2)\leq d(j,o_2)$,
and hence
\[
d(w_2,o_2)
\leq d(j,w_2)+d(j,o_2)
\leq 2\cdot d(j,o_2)
\leq 2\cdot\MC(\ov).
\]
Therefore, for every $a\in N_1$,
\[
\cost_a(\w)
=d(a,w_1)
\leq d(a,o_1)+d(w_1,o_1)
\leq 3\cdot\MC(\ov),
\]
and, similarly, for every $a\in N_2$,
\[
\cost_a(\w)
=d(a,w_2)
\leq d(a,o_2)+d(w_2,o_2)
\leq 3\cdot\MC(\ov).
\]
Thus, $\MC(\w)\leq 3\cdot\MC(\ov)$.

\item $w_1=o_2$. Since $w_1=t(i)$, agent $i$ is weakly closer to $o_2$ than to $o_1$, and therefore
\[
d(o_1,o_2)
\leq d(i,o_1)+d(i,o_2)
\leq 2\cdot d(i,o_1)
\leq 2\cdot\MC(\ov).
\]
For every agent $a\in N_1$,
\[
\cost_a(\w)
=d(a,o_2)
\leq d(a,o_1)+d(o_1,o_2)
\leq 3\cdot\MC(\ov).
\]

Since $o_1\neq o_2=w_1$, the candidate location $o_1$ is available when $F_2$ is placed. Hence,
\[
d(j,w_2)
\leq d(j,o_1)
\leq d(j,o_2)+d(o_1,o_2)
\leq 3\cdot\MC(\ov).
\]
By another application of the triangle inequality,
\[
d(w_2,o_2)
\leq d(j,w_2)+d(j,o_2)
\leq 4\cdot\MC(\ov).
\]
Thus, for every agent $a\in N_2$,
\[
\cost_a(\w)
=d(a,w_2)
\leq d(a,o_2)+d(o_2,w_2)
\leq 5\cdot\MC(\ov).
\]
Therefore, $\MC(\w)\leq 5\cdot\MC(\ov)$ in this subcase.
\end{itemize}

\medskip
\noindent
\emph{Case 3: all agents approve only $F_k$ for some $k\in[2]$.}
Only the location of $F_k$ affects the max cost. Let $i$ be the fixed agent that determines the outcome. Since $w_k=t(i)$, $d(i,w_k)\leq d(i,o_k)$, and thus
\[
d(w_k,o_k)
\leq d(i,w_k)+d(i,o_k)
\leq 2\cdot d(i,o_k)
\leq 2\cdot\MC(\ov).
\]
Consequently, for every agent $a$,
\[
\cost_a(\w)
=d(a,w_k)
\leq d(a,o_k)+d(w_k,o_k)
\leq 3\cdot\MC(\ov).
\]
Hence, $\MC(\w)\leq 3\cdot\MC(\ov)$.

The approximation ratio is at most $5$ in all cases, and the proof is complete.
\end{proof}

We next show that the upper bound of $5$ is tight. 

\begin{theorem}\label{thm:gmclower}
On general metric spaces, the approximation ratio of any deterministic strategyproof mechanism in terms of the max cost is at least $5-\delta$, for any $\delta>0$.
\end{theorem}

\begin{proof}
Suppose towards a contradiction that there is a deterministic strategyproof mechanism $M$ with approximation ratio at most $\rho=5-\delta$ for some $\delta>0$. Choose
\[
0<\varepsilon<\min\{1/4,\delta/3\}.
\]
We start with a tree metric with candidate locations $C=\{q,g,z_1,z_2\}$ and three branches
\[
\begin{aligned}
q&\;\overset{1-\varepsilon}{\rule{1.5em}{0.4pt}}\;h
  \;\overset{1}{\rule{1.5em}{0.4pt}}\;g,\\
q&\;\overset{1}{\rule{1.5em}{0.4pt}}\;y_1
  \;\overset{3-2\varepsilon}{\rule{1.5em}{0.4pt}}\;z_1,\\
q&\;\overset{1}{\rule{1.5em}{0.4pt}}\;y_2
  \;\overset{3-3\varepsilon}{\rule{1.5em}{0.4pt}}\;z_2,
\end{aligned}
\]
where $h,y_1,y_2$ will be used as agent positions. All seven points induce strict rankings of $C$. By \cref{lem:rankings}, we extend the metric so that every strict ranking of $C$ is induced by some point, without changing the distances between the points above.

Let $N_1=\{i_1,i_2\}$ and $N_2=\{j_1,j_2\}$, with all agents having singleton preferences. We first identify a facility whose induced range we will show excludes $q$. Consider the profile in which all four agents are positioned at $h$. Since the facilities must be placed at distinct locations, at least one of them is not placed at $q$; without loss of generality, assume that this is $F_2$. We fix the agents of $N_1$ at $h$ and let $\mathcal R$ denote the range of the induced rule for $F_2$ as the agents of $N_2$ vary over positions inducing strict rankings of $C$.

We first show that $\{g,z_1,z_2\}\subseteq\mathcal R$. Fix $c\in\{g,z_1,z_2\}$, and consider the instance in which the agents of $N_1$ are positioned at $q$ and those of $N_2$ at $c$. The optimal max cost is $0$, so $M$ must return $(q,c)$. Move the agents of $N_1$ one at a time from $q$ to $h$. By strategyproofness, $F_1$ remains at $q$: otherwise, the moving agent at her new position $h$ could misreport $q$ and obtain her unique closest candidate location.

In the resulting instance, the solution $(q,c)$ has max cost $1-\varepsilon$. If $F_2$ were placed at a location other than $c$, then, since $q$ is occupied by $F_1$, the agents of $N_2$ would incur cost at least
\[
\min_{\substack{a,b\in\{g,z_1,z_2\}\\a\ne b}}d(a,b)
=
6-4\varepsilon
>
\rho(1-\varepsilon),
\]
contradicting the approximation guarantee. Hence, $F_2$ is placed at $c$, proving the range claim.

By \cref{lem:dictator}, there is a dictator $j_d\in N_2$ whose closest candidate location in $\mathcal R$ is selected by the induced rule. When all four agents are positioned at $h$, this dictator ranks $q$ first, but $F_2$ is not placed at $q$. Therefore, $q\notin\mathcal R$, and thus
\[
\mathcal R=\{g,z_1,z_2\}.
\]

Finally, keep the agents of $N_1$ at $h$ and position $j_r$ at $y_r$ for each $r\in\{1,2\}$. The closest candidate location in $\mathcal R$ to the dictator $j_d$ is $z_d$, so $F_2$ is placed there. The other agent of $N_2$ then incurs cost
\[
d(y_{3-d},z_d)\ge 5-3\varepsilon>\rho.
\]
However, the solution $(g,q)$ has max cost $1$, contradicting the approximation guarantee of $M$.
\end{proof}

By \cref{thm:fixed-dictator-sp,thm:gmcupper,thm:gmclower}, we obtain the following result.

\begin{corollary}
On general metric spaces, the best possible approximation ratio of deterministic strategyproof mechanisms in terms of the max cost is exactly $5$.
\end{corollary}

\subsubsection{Universally Strategyproof Mechanisms}\label{subsec:metric-usp-mc}

We now consider universally strategyproof mechanisms for the max cost. We first present a mechanism with approximation ratio $19/5=3.8$, and then prove a lower bound of $15/4=3.75$, even for singleton instances.

\bigskip
\noindent
{\bf Upper bound.}
If $N_1\cap N_2\neq\varnothing$, or if exactly one of $N_1,N_2$ is non-empty, we run the corresponding case of \Mech{Fixed-Dictator}, whose approximation ratio is at most $3$ by Cases~1 and~3 in the proof of \cref{thm:gmcupper}. Hence, in the remainder, we consider singleton instances with $N_1,N_2\neq\varnothing$.

Fix an arbitrary agent $i_j\in N_j$ for every $j\in[2]$, independently of the reported positions. We randomize over the following two component mechanisms:
\begin{itemize}
\item \Mech{Serial-Dictator}. Choose $j\in[2]$ uniformly at random,
place $F_j$ at $t(i_j)$, and then place $F_{3-j}$ at the closest candidate
location to $i_{3-j}$ in $C\setminus\{t(i_j)\}$.

\item $\Mech{Guarded-Serial}_\theta$. Fix $\theta\in[1,4/3]$.
Choose $j\in[2]$ uniformly at random and place $F_j$ at $t(i_j)$. If
\[
d(i_j,s(i_j))>\frac{\theta}{2}\bigl(d(i_j,t(i_j))+d(i_j,s(i_j))\bigr),
\]
then place $F_{3-j}$ at the closest candidate location to $i_{3-j}$ in
$C\setminus\{t(i_j)\}$. Otherwise, place $F_{3-j}$ at the closest
candidate location to $i_{3-j}$ in
\[
\{s(i_j)\}\cup
\left\{
c\in C\setminus\{t(i_j)\}:
d(t(i_j),c)>(1+\theta)\bigl(d(i_j,t(i_j))+d(i_j,s(i_j))\bigr)
\right\}.
\]
\end{itemize}
We run \Mech{Serial-Dictator} with probability $4/5$. With the remaining probability $1/5$, we choose $\theta$ uniformly at random from $[1,4/3]$ and run $\Mech{Guarded-Serial}_\theta$. See \cref{mech:guarded-fixed-dictator}.

\begin{algorithm}[t]
\caption{\Mech{Guarded-Fixed-Dictator}}
\label{mech:guarded-fixed-dictator}
\KwIn{Reported positions of agents}
\KwOut{Facility locations $\w=(w_1,w_2)$}
\eIf{$N_1\cap N_2\neq\varnothing$, or exactly one of $N_1,N_2$ is non-empty}{
    Run the corresponding case of \Mech{Fixed-Dictator}\;
}{
    With probability $4/5$, run \Mech{Serial-Dictator}\;
    With probability $1/5$, sample $\theta$ uniformly from $[1,4/3]$ and run $\Mech{Guarded-Serial}_\theta$\;
}
\end{algorithm}

\begin{theorem}\label{thm:max-general-usp}
\Mech{Guarded-Fixed-Dictator} is universally strategyproof.
\end{theorem}

\begin{proof}
We show that every deterministic realization of \Mech{Guarded-Fixed-Dictator} is strategyproof. Fix the component mechanism, the facility placed first, and, for $\Mech{Guarded-Serial}_\theta$, the value of $\theta$.
The first branch uses \Mech{Fixed-Dictator}, which is strategyproof by \cref{thm:fixed-dictator-sp}. In the singleton branch, the first fixed agent approves only the facility placed first, which is located at her closest candidate location. Although her report may affect the candidate set available in the second stage, she does not care about the location of the other facility. For the second fixed agent, the location of the first facility, and hence the set of available candidate locations, is independent of her report. The mechanism places her approved facility at her closest available candidate location. Thus, neither fixed agent can benefit from misreporting, while all other agents are ignored.
\end{proof}

We next bound the approximation ratio. Let $\ov=(o_1,o_2)$ be an optimal solution. If $\MC(\ov)=0$, then $i_j$ is positioned at $o_j=t(i_j)$ for each $j\in[2]$. Since $\theta<2$, the unrestricted candidate set is always used, and both \Mech{Serial-Dictator} and $\Mech{Guarded-Serial}_\theta$ return an optimal solution.
Hence, we may assume that $\MC(\ov)>0$ and rescale all distances so that $\MC(\ov)=1$.
Then, for each $j\in[2]$, 
\[
d(i_j,t(i_j))\le d(i_j,o_j)\le1
\]
and
\[
d(t(i_j),s(i_j))\le d(i_j,t(i_j))+d(i_j,s(i_j)).
\]

We use the following bound in the analysis.

\begin{lemma}\label{lem:max-removal-probability}
Fix $j\in[2]$ and $c\in C$. Let $p$ be the probability, over the uniform
choice of $\theta\in[1,4/3]$, that the restricted candidate set is used
when $F_j$ is placed first and does not contain $c$. If $p>0$, then
\[
d(i_j,t(i_j))+d(i_j,s(i_j))\le3-p, \ 
d(i_j,s(i_j))\le2-p, \text{ and } \
d(t(i_j),c)\le7-3p.
\]
\end{lemma}

\begin{proof}
Since $s(i_j)$ is always retained, $p>0$ implies $c\neq s(i_j)$. The candidate location $c$ is removed exactly when
\[
2d(i_j,s(i_j)) \le\theta\bigl(d(i_j,t(i_j))+d(i_j,s(i_j))\bigr)
\]
and
\[
d(t(i_j),c) \le(1+\theta)\bigl(d(i_j,t(i_j))+d(i_j,s(i_j))\bigr).
\]
Both conditions become weaker as $\theta$ increases. Hence, if $c$ is removed with positive probability, it is removed exactly for $\theta\in[\ell,4/3]$ for some $\ell\in[1,4/3]$. Since $\theta$ is uniform on an interval of length $1/3$, $p=4-3\ell$, and thus $\ell=\frac{4-p}{3}$.
When $c=t(i_j)$, the second condition is automatic, so $p$ is simply the probability that the restricted candidate set is used.

At $\theta=\ell$, the first condition gives
\[
(2-\ell)d(i_j,s(i_j))
\le \ell d(i_j,t(i_j))
\le \ell,
\]
and therefore
\[
d(i_j,s(i_j))
\le \frac{\ell}{2-\ell}
= \frac{4-p}{2+p}
\le 2-p.
\]
Together with $d(i_j,t(i_j))\le1$, this yields
\[
d(i_j,t(i_j))+d(i_j,s(i_j))
\le \frac{2}{2-\ell}
= \frac6{2+p}
\le 3-p.
\]
Evaluating the second condition at $\theta=\ell$ and using the same bound gives us
\[
d(t(i_j),c)
\le \frac{2(1+\ell)}{2-\ell}
= \frac{2(7-p)}{2+p}
\le 7-3p. \qedhere
\]
\end{proof}

We are now ready to prove the upper bound. 

\begin{theorem}\label{thm:max-general-upper-19-5}
On general metric spaces, the approximation ratio of
\Mech{Guarded-Fixed-Dictator} in terms of the max cost is at most
$19/5$.
\end{theorem}

\begin{proof}
Consider a singleton instance with $N_1,N_2\neq\varnothing$ such that $\MC(\ov)=1$.
We first prove a bound that will be used for both component mechanisms. Suppose that $F_j$ is placed first at $t(i_j)$ and $F_{3-j}$ is subsequently placed at $w_{3-j}$. Every agent $a\in N_j$ has cost at most $3$, since
\[
d(a,t(i_j))
\le d(a,o_j)+d(o_j,i_j)+d(i_j,t(i_j))
\le3.
\]
If a candidate location $c$ remains available to the second fixed agent, then $d(i_{3-j},w_{3-j})\le d(i_{3-j},c)$
and hence, for every $a\in N_{3-j}$,
\[
d(a,w_{3-j})
\le d(a,o_{3-j})+d(o_{3-j},i_{3-j})+d(i_{3-j},c)
\le2+d(i_{3-j},c).
\]
Therefore,
\begin{equation}\label{eq:max-available-candidate}
\MC(\w)\le \max\{3,\,2+d(i_{3-j},c)\}.
\end{equation}
We now make some useful observations for the max cost of the component mechanisms. 
\begin{itemize}
\item \Mech{Serial-Dictator}. If either $t(i_{3-j})$ or $o_{3-j}$ remains available after placing the first facility at $t(i_j)$, then \eqref{eq:max-available-candidate} gives max cost at most $3$. Thus, the max cost can exceed $3$ only if both locations coincide with the unavailable candidate, that is,
\[
t(i_1)=t(i_2)=o_{3-j}.
\]
Since $o_1\neq o_2$, the conditions $t(i_1)=t(i_2)=o_2$ and $t(i_1)=t(i_2)=o_1$ cannot both hold. In particular, if $t(i_1)\neq t(i_2)$ or $t(i_1)=t(i_2)\notin\{o_1,o_2\}$, then the max cost of \Mech{Serial-Dictator} is at most $3$ in both orders.

\item $\Mech{Guarded-Serial}_\theta$. 
For $j\in[2]$, let
\[
E_j = \E\!\left[ \max\{0,\MC(\Mech{Guarded-Serial}_\theta)-3\} \,\middle|\, F_j\text{ is placed first} \right],
\]
where the expectation is over $\theta$.
Since each facility is placed first with probability $1/2$,
\begin{equation}\label{eq:max-excess-average}
\E[\MC(\Mech{Guarded-Serial}_\theta)]
\le 3+\frac{E_1+E_2}{2}.
\end{equation}
\end{itemize}
Since \Mech{Guarded-Fixed-Dictator} runs \Mech{Serial-Dictator} and $\Mech{Guarded-Serial}_\theta$ with probabilities $4/5$ and $1/5$, respectively, to prove the upper bound $19/5$, it suffices to show that
\[
4\E[\MC(\Mech{Serial-Dictator})]
+
\E[\MC(\Mech{Guarded-Serial}_\theta)]
\le19.
\]
The remaining analysis depends on whether \Mech{Serial-Dictator} can have max cost greater than $3$. As shown above, this is possible only when $t(i_1)=t(i_2)\in\{o_1,o_2\}$.

\medskip
\noindent
{\bf Case 1: $t(i_1)=t(i_2)\in\{o_1,o_2\}$.}
Without loss of generality, suppose that $t(i_1)=t(i_2)=o_2$.
Since $o_1\neq t(i_1)$ and $s(i_1)$ is the second-closest candidate location to $i_1$,
$d(i_1,s(i_1))\le d(i_1,o_1)\le1$ and 
$d(i_1,t(i_1))+d(i_1,s(i_1))\le2$.
When $F_1$ is placed first, $s(i_1)$ remains available for every realization of $\Mech{Guarded-Serial}_\theta$, and
\[
d(i_2,s(i_1))
\le
d(i_2,o_2)+d(o_2,s(i_1))
\le
1+d(i_1,t(i_1))+d(i_1,s(i_1)).
\]
Using $c=s(i_1)$ in \eqref{eq:max-available-candidate}, we obtain
\[
\MC(\Mech{Guarded-Serial}_\theta)
\le \max\{3,2+d(i_2,s(i_1))\} 
\le 3+d(i_1,t(i_1))+d(i_1,s(i_1)).
\]
Therefore, for every value of $\theta$, the excess over $3$ is  at most $d(i_1,t(i_1))+d(i_1,s(i_1))$. Taking expectation over $\theta$, we have 
\[
E_1 \le d(i_1,t(i_1))+d(i_1,s(i_1)) \le2.
\]

Let $q$ be the probability that the restricted candidate set is used when $F_2$ is placed first. Without the restriction, $o_1$ remains available and the max cost is at most $3$. With the restriction, $s(i_2)$ remains available, and, since
$o_2=t(i_1)=t(i_2)$,
\[
d(i_1,s(i_2))
\le d(i_1,o_2)+d(o_2,s(i_2))
\le 1+d(i_2,t(i_2))+d(i_2,s(i_2)).
\]
Therefore,
\[
E_2
\le q\bigl(d(i_2,t(i_2))+d(i_2,s(i_2))\bigr).
\]
The value of $q$ also determines how strongly the geometry constrains \Mech{Serial-Dictator}, so we consider the two possibilities separately.

\begin{itemize}
\item $q>0$. By \cref{lem:max-removal-probability} with $c=t(i_2)$, $d(i_2,t(i_2))+d(i_2,s(i_2))\le3-q$ and $d(i_2,s(i_2))\le2-q$. Thus, $E_2\le q(3-q)\le2$, and \eqref{eq:max-excess-average} gives
\[
\E[\MC(\Mech{Guarded-Serial}_\theta)] \le 3+\frac{2+2}{2} = 5.
\]
For \Mech{Serial-Dictator}, when $F_1$ is placed first, the unavailable candidate location is $t(i_1)=t(i_2)$, so $s(i_2)$ remains available. By \eqref{eq:max-available-candidate},
\[
\MC(\Mech{Serial-Dictator}) \le \max\{3,\,2+d(i_2,s(i_2))\} \le4.
\]
In the opposite order, $o_1$ remains available, and thus the max cost is at most $3$. Hence,
\[
\E[\MC(\Mech{Serial-Dictator})] \le \frac{4+3}{2} = \frac72.
\]
Therefore,
\[
4\E[\MC(\Mech{Serial-Dictator})] + \E[\MC(\Mech{Guarded-Serial}_\theta)] \le19.
\]

\item $q=0$. Hence, $E_2=0$.
Consider now the order in which $F_1$ is placed first. 
The restricted candidate set is used when
\[
d(i_1,s(i_1)) \leq \frac12 \theta\bigl(d(i_1,t(i_1))+d(i_1,s(i_1))\bigr).
\]
Since $d(i_1,s(i_1))\le 1$, this condition holds whenever $\theta\bigl(d(i_1,t(i_1))+d(i_1,s(i_1))\bigr) \ge 2$.
In this case, every retained candidate location $c\neq s(i_1)$ satisfies
\[
d(o_2,c)
> (1+\theta)\bigl(d(i_1,t(i_1))+d(i_1,s(i_1))\bigr)
\ge d(i_1,t(i_1))+d(i_1,s(i_1))+2.
\]
On the other hand,
\[
d(i_2,s(i_1))
\le d(i_2,o_2)+d(o_2,s(i_1))
\le 1+d(i_1,t(i_1))+d(i_1,s(i_1)),
\]
whereas
\[
d(i_2,c)
\ge d(o_2,c)-d(i_2,o_2)
> d(i_1,t(i_1))+d(i_1,s(i_1))+1.
\]
Thus, the second fixed agent chooses $s(i_1)$. Since
\[
d(o_2,s(i_1))
\le
d(i_1,t(i_1))+d(i_1,s(i_1))
\le2,
\]
the max cost is at most $3$, and thus the excess is zero whenever $\theta\bigl(d(i_1,t(i_1))+d(i_1,s(i_1))\bigr) \ge 2$. Otherwise, since $s(i_1)$ is available, \eqref{eq:max-available-candidate} and the bound above on $d(i_2,s(i_1))$ show that the excess is always at most
\[
d(i_1,t(i_1))+d(i_1,s(i_1)).
\]
Since
$0<d(i_1,t(i_1))+d(i_1,s(i_1))\le2$
and
$\Pr[\theta<x]\le3(x-1)$ for any $x\ge1$, 
\[
\Pr\!\left[ \theta\bigl(d(i_1,t(i_1))+d(i_1,s(i_1))\bigr)<2 \right]
\le 3\left( \frac{2}{d(i_1,t(i_1))+d(i_1,s(i_1))}-1 \right).
\]
Therefore,
\[
E_1 \le 3\bigl( 2-d(i_1,t(i_1))-d(i_1,s(i_1)) \bigr).
\]

For \Mech{Serial-Dictator}, when $F_1$ is placed first, $s(i_1)$ remains available. 
By \eqref{eq:max-available-candidate},
\[
\MC(\Mech{Serial-Dictator})
\le 3+d(i_1,t(i_1))+d(i_1,s(i_1)).
\]
In the opposite order, $o_1$ remains available, and thus the max cost is at most $3$. Hence,
\[
\E[\MC(\Mech{Serial-Dictator})]
\le 3+\frac{d(i_1,t(i_1))+d(i_1,s(i_1))}{2}.
\]
Together with \eqref{eq:max-excess-average} and $E_2=0$,
\begin{align*}
&4\E[\MC(\Mech{Serial-Dictator})]
+\E[\MC(\Mech{Guarded-Serial}_\theta)] \\
&\le
4\left( 3+\frac{d(i_1,t(i_1))+d(i_1,s(i_1))}{2} \right) +3
+\frac32\bigl( 2-d(i_1,t(i_1))-d(i_1,s(i_1)) \bigr)\\
&= 18+\frac{d(i_1,t(i_1))+d(i_1,s(i_1))}{2}\\
&\le19.
\end{align*}
\end{itemize}

\medskip
\noindent
{\bf Case 2: $t(i_1)\neq t(i_2)$ or $t(i_1)=t(i_2)\notin\{o_1,o_2\}$.}
In this case, \Mech{Serial-Dictator} has max cost at most $3$ in both orders. It therefore suffices to prove
\[
\E[\MC(\Mech{Guarded-Serial}_\theta)]\le7,
\]
or, by \eqref{eq:max-excess-average},
\[
E_1+E_2\le8.
\]

We first show that $E_j\le6$ for any $j \in [2]$. 
Let $q$ be the probability that the restricted candidate set is used and does not contain $o_{3-j}$. 
If $q=0$, the excess is $0$: either the unrestricted rule is used, in which case the outcome is that of \Mech{Serial-Dictator} in the same order, or $o_{3-j}$ remains available. Otherwise, when $q > 0$, the protected candidate location $s(i_j)$ remains available, and
\[
d(i_{3-j},s(i_j)) \le d(i_{3-j},o_{3-j}) +d(o_{3-j},t(i_j)) +d(t(i_j),s(i_j)).
\]
Hence,
\[
d(i_{3-j},s(i_j))-1
\le d(t(i_j),o_{3-j}) +d(i_j,t(i_j)) +d(i_j,s(i_j)).
\]
Using \cref{lem:max-removal-probability} with $c=o_{3-j}$, we have
\[
d(t(i_j),o_{3-j})\le7-3q
\]
and
\[
d(i_j,t(i_j))+d(i_j,s(i_j))\le3-q.
\]
Therefore,
\begin{equation}\label{eq:max-uniform-excess}
E_j \le q(10-4q) \le 6.
\end{equation}

If either $E_1\le2$ or $E_2\le2$, then we immediately have that $E_1+E_2\le8$. 
Hence, assume $E_1>2$ and $E_2>2$. 
In particular, the restricted candidate set is used with positive probability in both orders. 
By its activation condition,
\[
d(i_j,s(i_j)) \le \frac{\theta}{2}\bigl(d(i_j,t(i_j))+d(i_j,s(i_j))\bigr).
\]
Since $\theta\le4/3$, we have $d(i_j,s(i_j))\le2d(i_j,t(i_j))\le2$ for each $j\in[2]$.

Now let $p_j$ be the probability that the restricted candidate set is used when $F_j$ is placed first and does not contain $s(i_{3-j})$. Outside this event, either the unrestricted rule is used, in which case the excess is zero, or $s(i_{3-j})$ remains available. In the latter case, \eqref{eq:max-available-candidate} bounds the excess by
\[
\max\{0,d(i_{3-j},s(i_{3-j}))-1\}.
\]
On the event defining $p_j$, the protected candidate location $s(i_j)$ remains available. Since
\[
d(i_{3-j},s(i_j))
\le d(i_{3-j},t(i_{3-j})) +d(t(i_{3-j}),s(i_j))
\le 1+d(t(i_{3-j}),s(i_j)),
\]
the excess is at most $d(t(i_{3-j}),s(i_j))$. Thus, 
\begin{equation}\label{eq:max-coupled-removals}
E_j \le \max\{0,d(i_{3-j},s(i_{3-j}))-1\} + p_jd(t(i_{3-j}),s(i_j)).
\end{equation}
Since $E_j>2$ and $d(i_{3-j},s(i_{3-j}))\le2$, the first term is at most $1$, and hence $p_j>0$ for both $j$.
Applying \cref{lem:max-removal-probability} in the opposite order, with first facility $F_{3-j}$ and candidate location $c=s(i_j)$, we obtain
\[
d(i_{3-j},s(i_{3-j})) \le 2-p_{3-j}
\]
and 
\[
d(t(i_{3-j}),s(i_j)) \le 7-3p_{3-j}.
\]
Substituting these bounds into \eqref{eq:max-coupled-removals},
\[
E_j \le 1-p_{3-j} + p_j(7-3p_{3-j}).
\]
Summing over the two orders,
\[
E_1+E_2
\le 2+6(p_1+p_2)-6p_1p_2
= 8-6(1-p_1)(1-p_2) 
\le8.
\]
This completes the proof.
\end{proof}

\bigskip
\noindent
{\bf Lower bound.}
We now show that no universally strategyproof mechanism can achieve an approximation ratio better than $15/4$, even on singleton instances. We construct a metric space consisting of $k$ clusters of candidate locations, together with two types of profiles for each group of agents. In every instance we use, the optimal max cost is close to $1$. We first show that the max cost is typically at least $3$. We then identify two families of instances: in the first, the max cost is close to $5$ with a probability that decreases with a parameter $\mu$, while in the second it is close to $7$ with a probability that increases with $\mu$. Combining the two families yields the bound.

\medskip
\noindent
\emph{The metric space.}
Fix an integer $k\geq5$ and let $\varepsilon=1/k$. We construct the metric space as follows.
\begin{itemize}
    \item For each $j\in[k]$, cluster $j$ contains the candidate locations $q_j$ and
    $B_j=\{b_{j1},\ldots,b_{jk}\}$, together with two points $a_j$ and $u_j$ that will be used as agent positions.

    \item Within cluster $j$, we connect 
    $a_j$ to $q_j$ by an edge of length $1$,
    $a_j$ to each $b_{ji}$ by an edge of length $1+\varepsilon$, 
    and $q_j$ to $u_j$ by an edge of length $1$.

    \item We attach a leaf at distance $1$ to $q_j$ and to each $b_{ji}$. These leaves will be used as the positions of exceptional agents.

    \item For every two distinct clusters $j$ and $h$, we connect $u_j$ to $q_h$ and to every location of $B_h$ by an edge of length $3-\varepsilon$.

    \item We add three auxiliary candidate locations, each attached to $a_1$ by an edge of length $k^4$. These locations are far from every agent position we use. Their only role is to guarantee that the ranges of the single-facility rules induced by a mechanism contain at least three locations, so that \cref{lem:dictator} applies.
\end{itemize}
The resulting weighted graph is connected and has strictly positive edge lengths, so its shortest-path distances define a metric.

To apply \cref{lem:dictator}, we need all agent positions used below to induce strict rankings of the candidate locations. We therefore break potential ties in the construction as follows. We replace every candidate location $c$ by a new candidate location attached to $c$ by an edge of length $\delta_c\in(0,\varepsilon^2)$, where all $\delta_c$ are distinct. The point $c$ is no longer a candidate location, and we use the name $c$ for the new one. Since all distances before this modification are multiples of $\varepsilon$, this breaks every tie without reversing any strict preference, and changes every relevant distance by less than $\varepsilon^2$. Finally, by \cref{lem:rankings}, we add points so that every strict ranking of the candidate set $C$ is induced by some point, without changing the distances between the points above.

\medskip
\noindent
\emph{The profiles.}
There are two groups of $k$ agents, where the agents of group $r$ approve only facility $F_r$, for each $r\in[2]$. For $j,i\in[k]$ and each choice of an exceptional agent in a group, we define the following profiles.

\begin{itemize}
    \item In $L_{ji}$, the $k-1$ non-exceptional agents are positioned at $a_j$, while the exceptional agent is positioned at the leaf attached to $b_{ji}$.

    \item In $H_j$, the $k-1$ non-exceptional agents are positioned at $u_j$, while the exceptional agent is positioned at the leaf attached to $q_j$.
\end{itemize}
Thus, the identity of the exceptional agent is regarded as part of the profile.
Ignoring the perturbations of less than $\varepsilon^2$, the max cost of a group when its facility is placed at each location is as follows, where $h\neq j$:
\[
\begin{array}{c|ccccc}
 &q_j&b_{ji}&B_j\setminus\{b_{ji}\}&q_h&B_h\\ \hline
L_{ji}&3+\varepsilon&1+\varepsilon&3+2\varepsilon&5-\varepsilon&7-2\varepsilon\\
H_j&1&3+\varepsilon&3+\varepsilon&5-\varepsilon&5-\varepsilon.
\end{array}
\]

Every auxiliary location gives max cost at least $k^4$. In particular, since the two facilities can be placed at $q_j$ and $b_{ji}$, or at $b_{ji}$ and $b_{h\ell}$, every instance of the form $(H_j,L_{ji})$, $(L_{ji},H_j)$, or $(L_{ji},L_{h\ell})$ with $h\neq j$ has optimal max cost at most $1+\varepsilon+\varepsilon^2$.

Since any mechanism with approximation ratio greater than $4$ already satisfies our desired lower bound, it suffices to consider mechanisms with approximation ratio at most $4$. We first show that such a mechanism induces dictatorial single-facility rules with high probability.

\begin{lemma}\label{lem:max-induced-dictatorship}
Consider a universally strategyproof mechanism with approximation ratio at most $4$. With probability $1-O(1/k)$ over the random choice of a deterministic mechanism, the following holds: whenever the agents of one group have profile $H_j$ or $L_{ji}$, for any $j,i\in[k]$ and any choice of the exceptional agent, the single-facility rule induced for the other group is dictatorial on a range that contains all three auxiliary locations.
\end{lemma}

\begin{proof}
Fix the profile $T$ of one group, including the identity of its exceptional agent, and consider the instance in which all agents of the other group are positioned at an auxiliary location $c$. The optimal max cost is at most $1+\varepsilon+\varepsilon^2$. If the range of the induced rule does not contain $c$, then the max cost is at least $k^4$. Since the approximation ratio is at most $4$, the probability of choosing a deterministic mechanism whose induced rule does not contain $c$ in its range is at most $4(1+\varepsilon+\varepsilon^2)/k^4=O(k^{-4})$.

There are $k^2(k+1)$ possible profiles $T$ for each group and three auxiliary locations. Hence, by the union bound, all the required ranges contain the three auxiliary locations with probability $1-O(1/k)$.

After the tie-breaking modification, every agent position we use induces a strict ranking of the candidate locations, and every strict ranking is induced by some point. Since each of the ranges above contains at least three locations, \cref{lem:dictator} shows that each induced rule is dictatorial on its range.
\end{proof}

\begin{theorem}\label{thm:max-general-lower-15-4}
On general metric spaces, the approximation ratio of every universally strategyproof mechanism in terms of the max cost is at least $15/4-\delta$, for any $\delta>0$, even on singleton instances.
\end{theorem}

\begin{proof}
Consider any universally strategyproof mechanism with approximation ratio at most $4$. By \cref{lem:max-induced-dictatorship}, with probability $1-O(1/k)$ over the choice of a deterministic mechanism in its support, all the relevant induced single-facility rules are dictatorial on ranges containing the three auxiliary locations. We therefore fix such a deterministic mechanism $M$ for the remainder of the argument.

We call a candidate location \emph{original} if it is not auxiliary. For $r\in[2]$, let $\alpha_r$ be the probability that, when the agents that do not approve $F_r$ have profile $L_{JI}$, where $J$, $I$, and the exceptional agent are chosen uniformly at random, all original locations in the range of the rule for $F_r$ belong to $B_J$. Let
\[
\mu=\min\{1,\alpha_1+\alpha_2\}.
\]

We first establish a bound that will be used for both families of instances. Consider a group with profile $L_{JI}$. The rule for its facility, including its dictator, is determined by the profile of the other group. Since the exceptional agent is chosen uniformly at random, she is the dictator with probability $1/k$. Otherwise, the dictator is positioned at $a_J$, and the location she chooses is independent of $I$. Since $I$ is uniform in $[k]$, this location is $b_{JI}$ with probability at most $1/k$. Every other location gives the group max cost at least $3$. Hence, in either family below, the max cost is at least $3$ with probability at least $1-2/k$.

\medskip
\noindent
\emph{First family.}
The first family consists of the instances $(H_J,L_{JI})$ and $(L_{JI},H_J)$, where $J$, $I$, the order of the two profiles, and the exceptional agents are chosen uniformly at random.

To analyze this family, consider first the instance $(L_{JI},L_{JI'})$, in which both groups are associated with the same cluster $J$. Unless one of the two exceptional agents is a dictator, the two ranges cannot both contain $q_J$; otherwise, both dictators would be positioned at $a_J$ and would choose their closest location $q_J$, so the two facilities would be placed at the same location. Hence, with probability at least $1-2/k$, at least one of the two ranges does not contain $q_J$.

By the definition of $\alpha_1$ and $\alpha_2$, the expected number of ranges whose original locations all belong to $B_J$ is $\alpha_1+\alpha_2$. Since every such range excludes $q_J$, the expected number of ranges that exclude $q_J$ and contain an original location outside $B_J$ is at least
\[
\max\left\{0,1-\alpha_1-\alpha_2-\frac{2}{k}\right\}
\ge
1-\mu-\frac{2}{k}.
\]
Therefore, the probability that a uniformly selected range does not contain $q_J$ and contains an original location outside $B_J$ is at least
\[
\frac{1-\mu}{2}-O(1/k).
\]

Replace the selected group's profile by $H_J$, sampling its exceptional agent independently. This produces the first-family distribution without changing that group's induced range or dictator.

Except with probability $1/k$, its dictator is positioned at $u_J$. Since $q_J$ is not in the range and the range contains an original location outside $B_J$, every location of another cluster is closer to $u_J$, at distance $3-\varepsilon$, than any location of $B_J$, at distance $3+\varepsilon$, or any auxiliary location. Hence, the dictator chooses a location of another cluster. By the table above, the max cost is then at least $5-\varepsilon$.

The max cost is at least $3$ except with probability $O(1/k)$, and is at least $5-\varepsilon$ on the event above. Separating this event from the remaining instances and using $\varepsilon=1/k$, we obtain
\begin{equation}\label{eq:short-lb-local-family}
\E[\MC\mid\text{first family}]
\ge
3+(2-\varepsilon)\frac{1-\mu}{2}-O(1/k)
=
4-\mu-O(1/k).
\end{equation}

\medskip
\noindent
\emph{Second family.}
The second family consists of the instances $(L_{JI},L_{J'I'})$, where $J$, $J'$, $I$, $I'$, and the exceptional agents are chosen uniformly at random, conditioned on $J\neq J'$.

Consider the event that all original locations in the range of the rule for $F_1$ belong to $B_{J'}$, and the event that all original locations in the range of the rule for $F_2$ belong to $B_J$. Before conditioning on $J\neq J'$, the first event depends only on the profile of the second group and occurs with probability $\alpha_1$, while the second depends only on the profile of the first group and occurs with probability $\alpha_2$. Hence, the two events are independent. Writing $s=\alpha_1+\alpha_2$, at least one of the two events occurs with probability
\[
\alpha_1+\alpha_2-\alpha_1\alpha_2
\ge
s-\frac{s^2}{4}
\ge
\frac34\min\{1,s\}
=
\frac34\mu.
\]
The second inequality follows directly by considering $s\le1$ and $s\ge1$. Conditioning on $J\neq J'$ changes this probability by $O(1/k)$.

If the first event occurs, then $F_1$ is placed either at a location of $B_{J'}$ or at an auxiliary location. Since the agents approving $F_1$ have profile $L_{JI}$ and $J\neq J'$, the table above shows that the max cost is at least $7-2\varepsilon$. The second event is symmetric.

Again, the max cost is at least $3$ except with probability $O(1/k)$, and is at least $7-2\varepsilon$ on the event above. Hence,
\begin{equation}\label{eq:short-lb-cross-family}
\E[\MC\mid\text{second family}]
\ge
3+(4-2\varepsilon)\frac{3\mu}{4}-O(1/k)
=
3+3\mu-O(1/k).
\end{equation}

\medskip
\noindent
\emph{Combining the two families.}
Choose an instance from the first family with probability $3/4$ and from the second family with probability $1/4$. By \eqref{eq:short-lb-local-family} and \eqref{eq:short-lb-cross-family}, the expected max cost of $M$ is at least
\[
\frac34(4-\mu)+\frac14(3+3\mu)-O(1/k)
=
\frac{15}{4}-O(1/k).
\]

The distribution over instances does not depend on $M$, and the bound holds for every deterministic mechanism satisfying the conclusion of \cref{lem:max-induced-dictatorship}, even though $\mu$ may depend on the mechanism. Averaging over the random choice of the deterministic mechanism and accounting for the ignored probability of $O(1/k)$, the expected max cost is still at least $15/4-O(1/k)$. Since the optimal max cost of every instance we use is at most $1+\varepsilon+\varepsilon^2$, some instance has approximation ratio at least
\[
\frac{15/4-O(1/k)}{1+1/k+1/k^2}
=
\frac{15}{4}-O(1/k).
\]

Choosing $k$ sufficiently large completes the proof.
\end{proof}

\section{Conclusion}\label{sec:conclusion}
In this paper, we revisited a truthful two-facility location problem with candidate locations and showed bounds on the best possible approximation ratio of deterministic and universally strategyproof mechanisms in terms of the social cost and the max cost. For deterministic mechanisms, we showed tight bounds both on the line and on general metric spaces. The difference between the two metric spaces is particularly significant for the social cost: the best possible approximation ratio is $7$ on the line and $2n-1$ on general metric spaces. Randomization leads to significant improvements. For universally strategyproof mechanisms, the best possible approximation ratio in terms of the social cost is between $5/2$ and $3$ on the line and tends to $4$ as the number of agents grows on general metric spaces. In terms of the max cost, the best possible approximation ratio is exactly $3$ on the line, while on general metric spaces it lies between $15/4$ and $19/5$.

The main remaining gaps for universally strategyproof mechanisms are therefore for the social cost on the line and the max cost on general metric spaces. It would also be interesting to determine the best possible approximation ratio in terms of the social cost on general metric spaces for every finite number of agents, to understand the power of mechanisms that are truthful in expectation, and to consider metric spaces that are more general than the line but more structured than general metric spaces, such as trees and median graphs.

\subsection*{AI Disclosure}
ChatGPT-6 Astra and ChatGPT-5.6 Sol were used during the development of this work to assist with formulating and verifying some of the theoretical arguments and proofs. The authors independently verified all mathematical claims and proofs presented as established results and take full responsibility for the content of the paper.

\bibliographystyle{plainnat}
\bibliography{twofl_refs}

\clearpage\newpage
\appendix

\section{Conjectured Tight \texorpdfstring{$5/2$}{5/2} Upper Bound on the Line}\label{app:line52}

In this section, we define the mechanism \Mech{Mixed-Percentile}.
Without loss of generality, we assume throughout that $n_1\ge n_2$; when $n_2>n_1$, the mechanism exchanges the labels of the two facilities. We break ties between candidate locations to the left, and, when two consecutive candidate pairs have the same total distance from an agent, we choose the left pair.

For $\lambda\ge0$, we define
\begin{equation}\label{line52:eq:F}
F_\lambda(q)=\frac{(1+2\lambda)q}{2(1+\lambda-q)}
\quad(0\le q\le1/2),\qquad
F_\lambda(q)=1-F_\lambda(1-q)
\quad(1/2\le q\le1).
\end{equation}
For a non-empty set $G$ of agents and a non-empty set $C'\subseteq C$ of candidate locations, the \emph{weighted-percentile rule} $W_{\lambda,G}(C')$ selects as \emph{representative} the $\ell$-th leftmost agent of $G$ with probability $F_\lambda(\ell/|G|)-F_\lambda((\ell-1)/|G|)$ and returns the candidate location in $C'$ closest to her reported position. These are weighted-percentile voting mechanisms in the sense of \citet{FFG}; for $\lambda=0$, the rule is their \Mech{Spike} mechanism.

The mechanism uses as components the following simpler mechanisms.
\begin{itemize}
\item $\Mech{Serial-Percentile}_1$, denoted by $T_1(\theta)$, places $F_1$ at $w_1=W_{\theta,N_1\setminus N_2}(C)$ and then places $F_2$ at $W_{0,N_2}(C\setminus\{w_1\})$. $\Mech{Serial-Percentile}_2$, denoted by $T_2(\theta)$, is defined symmetrically, placing $F_2$ first using the singleton group $N_2\setminus N_1$.

\item \Mech{Pair-Percentile}, denoted by $P_\lambda$, chooses a representative among the doubleton agents using $F_\lambda$, places the facilities at her two closest distinct candidate locations, and assigns the two locations to the facilities uniformly at random.
\end{itemize}

If all agents have doubleton preferences, the mechanism runs $P_0$. If only one facility is approved, it uses \Mech{Spike} for that facility and places the other facility at any available  candidate location. In the remaining cases, $n_1>0$ and both facilities are approved. Writing $[x]_+=\max\{x,0\}$, define
\begin{equation}\label{line52:eq:parameters}
\zeta=\left[1-\frac{4n_{12}}{n_1+n_2}\right]_+,
\qquad
\kappa=n_{12}+\frac{n_2^2\zeta}{n_2+16n_{12}}.
\end{equation}
The quantity $\zeta$ is positive exactly when $4n_{12}<n_1+n_2$. It determines both the adjustment to the serial weights and the common parameter $\kappa$ used to choose their first-stage percentiles.

\begin{algorithm}[t]
\caption{\Mech{Mixed-Percentile}}
\label{mech:percentile}
\KwIn{Reported positions of agents}
\KwOut{Facility locations $\w=(w_1,w_2)$}
\If{$n_1=n_2=0$}{
    Run $P_0$\;
}
\ElseIf{only $F_j$ is approved}{
    Place it at $W_{0,N_j}(C)$ and the other facility at any available candidate location\;
}
\Else{
    Compute $\zeta$ and $\kappa$ from \eqref{line52:eq:parameters}\;
    Run $T_1(\kappa/n_1)$, $T_2(\kappa/n_2)$, and $P_{(n_1+n_2)/(2n_{12})}$ with probabilities proportional to
    $3n_1$, $3n_2\left(1-\frac{n_1-n_2}{n_1+2n_2}\zeta\right)$, and $[4n_{12}-n_1-n_2]_+$, respectively\;
}
\end{algorithm}

Components with zero weight are omitted, and their parameters need not be evaluated. All weights are non-negative: $0\le\zeta\le1$ and
\[
1-\frac{n_1-n_2}{n_1+2n_2}\zeta
\ge
1-\frac{n_1-n_2}{n_1+2n_2}
=
\frac{3n_2}{n_1+2n_2}\ge0.
\]
The first weight is positive whenever the final branch is reached.

The mechanism has a single threshold, $4n_{12}=n_1+n_2$. Below this threshold, it uses only the two serial components. At and above the threshold, $\zeta=0$ and $\kappa=n_{12}$, so their weights are $3n_1$ and $3n_2$, while the pair component receives weight $4n_{12}-n_1-n_2$. In particular, at the threshold the pair component has zero weight and the two serial components are chosen in proportions $n_1:n_2$.

The first-stage parameters of the serial components satisfy
\[
n_1\left(\frac{\kappa}{n_1}\right)
=
n_2\left(\frac{\kappa}{n_2}\right)
=
\kappa\ge n_{12},
\]
whenever the corresponding singleton groups are non-empty. Thus, the distance term supplied by the weighted-percentile bound has sufficient coefficient to account for the additional cost incurred by all doubleton agents when the first facility moves away from its optimal location. Below the threshold, the correction in $\kappa$ provides an additional saving. Above the threshold, the pair component also represents the doubleton agents directly.

\begin{lemma}\label{line52:lem:usp}
\Mech{Mixed-Percentile} is universally strategyproof.
\end{lemma}

\begin{proof}
All randomization probabilities and percentile parameters depend only on $n_1,n_2,n_{12}$, and the random ranks used in different stages are independent. Hence, it suffices to fix all random choices and show that the resulting deterministic mechanism is strategyproof. Every fixed-rank rule $W_{\lambda,G}(C')$ is strategyproof: with the reports of the other agents fixed, the outcome is the reporting agent's closest candidate location clipped to an interval determined by those reports. Since her preferences over candidate locations are single-peaked, reporting truthfully yields her most preferred attainable outcome.

In $T_1(\theta)$ and $T_2(\theta)$, every agent considered in the first stage is a singleton agent and therefore cares only about the facility placed in that stage; although her report may change the candidate set available later, this cannot affect her own cost. For an agent considered in the second stage, the location chosen in the first stage, and hence the available candidate set, is fixed independently of her report. If such an agent is doubleton, the contribution of the first facility to her cost is likewise fixed. Thus, every deterministic realization of the serial components is strategyproof.

For $P_\lambda$, order the consecutive candidate pairs by their left endpoints. For adjacent pairs $\{c_j,c_{j+1}\}$ and $\{c_{j+1},c_{j+2}\}$, the difference between their costs to an agent at position $x$ is $d(x,c_j)-d(x,c_{j+2})$, which changes sign at $(c_j+c_{j+2})/2$. These thresholds increase with $j$, so preferences over consecutive pairs are single-peaked, and the most preferred pair is non-decreasing in the agent's position. Hence, each fixed-rank realization of $P_\lambda$ is strategyproof by the same fixed-rank argument. The random assignment of the two selected locations to the facilities does not affect the cost of a doubleton agent, while singleton reports are ignored.

Finally, when only one facility is approved, the fixed-rank argument applies to that facility, and the location of the unapproved facility does not affect any agent's cost.
\end{proof}

ChatGPT 6.0 Astra has produced a computer-generated proof candidate for this bound. However, its length and complexity make it prohibitively difficult to verify reliably by hand. We therefore do not claim the argument as a proof and state the bound below as a conjecture.

\begin{conjecture}
On the line, \textnormal{\Mech{Mixed-Percentile}} has approximation ratio at most $5/2$ in terms of the social cost.
\end{conjecture}

\end{document}